\documentclass[sigconf, nonacm, pdfa]{acmart}

\newcommand\vldbdoi{10.14778/3749646.3749692}
\newcommand\vldbpages{4269 - 4281}
\newcommand\vldbvolume{18}
\newcommand\vldbissue{11}
\newcommand\vldbyear{2025}
\newcommand\vldbauthors{\authors}
\newcommand\vldbtitle{\shorttitle}  
\newcommand\vldbavailabilityurl{URL_TO_YOUR_ARTIFACTS}
\newcommand\vldbpagestyle{empty} 

\usepackage{amsmath}
\usepackage{upgreek}  
\usepackage[ruled, vlined, linesnumbered]{algorithm2e}
\usepackage{multirow}
\usepackage{graphicx}
 \usepackage{colortbl}
\usepackage{tikz}
\usepackage{pgfplots}
\usetikzlibrary{patterns}
\usepgfplotslibrary{groupplots}
\usepackage{subfig}
\usepackage{enumitem}
\usepackage{xcolor}
\usepackage[normalem]{ulem}
\usepackage[a-2b]{pdfx}

\newenvironment{customlegend}[1][]{
    \begingroup
    \csname pgfplots@init@cleared@structures\endcsname
    \pgfplotsset{#1}%
}{%
    \csname pgfplots@createlegend\endcsname
    \endgroup
}%
\def\addlegendimage{\csname pgfplots@addlegendimage\endcsname}

\newtheorem{theorem}{Theorem}

\DeclareMathOperator{\argmin}{argmin}
\DeclareMathOperator{\vol}{vol}
\DeclareMathOperator{\tlog}{tlog}
\DeclareMathOperator{\knn}{KNN}
\DeclareMathOperator{\diag}{diag}
\DeclareMathOperator{\cs}{cosSim}

\DeclareMathOperator{\eig}{eigen}

\newcommand{\transpose}{^\texttt{T}\xspace}
\newcommand{\ewise}{^\circ\xspace}

\newcommand{\ie}{\textit{i.e.,}\xspace}
\newcommand{\eg}{\textit{e.g.,}\xspace}

\newcommand{\ahg}{H\xspace}

\newcommand{\HG}{{H}\xspace}
\newcommand{\HGA}{\mathcal{H}\xspace}

\newcommand{\VS}{\mathcal{V}\xspace}
\newcommand{\VSA}{\mathcal{V}_X\xspace}
\newcommand{\ES}{\mathcal{E}\xspace}
\newcommand{\ESK}{\mathcal{E}_K\xspace}
\newcommand{\ESA}{\mathcal{E}_X\xspace}
\newcommand{\LS}{\mathcal{L}\xspace}

\newcommand{\RN}{\mathbb{R}\xspace}

\newcommand{\HM}{\mathbf{H}\xspace}
\newcommand{\HMO}{\mathbf{H}_0\xspace}
\newcommand{\HMK}{\mathbf{H}_K\xspace}

\newcommand{\HL}{\mathbf{\widetilde{H}}\xspace}

\newcommand{\XM}{\mathbf{X}\xspace}
\newcommand{\DM}{\mathbf{D}\xspace}
\newcommand{\WM}{\mathbf{W}\xspace}
\newcommand{\PM}{\mathbf{P}\xspace}
\newcommand{\MM}{\mathbf{M}\xspace}
\newcommand{\ZM}{\mathbf{Z}\xspace}
\newcommand{\YM}{\mathbf{Y}\xspace}
\newcommand{\UM}{\mathbf{U}\xspace}
\newcommand{\VM}{\mathbf{V}\xspace}
\newcommand{\LM}{\mathbf{L}\xspace}
\newcommand{\IM}{\mathbf{I}\xspace}
\newcommand{\QM}{\mathbf{Q}\xspace}

\newcommand{\DtM}{\mathbf{\Delta}\xspace}
\newcommand{\LdM}{\mathbf{\Lambda}\xspace}
\newcommand{\PiM}{\mathbf{\Pi}\xspace}
\newcommand{\SgM}{\mathbf{\Sigma}\xspace}

\newcommand{\GmM}{\mathbf{\Gamma}\xspace}
\newcommand{\PsM}{\mathbf{\Psi}\xspace}

\newcommand{\ThM}{\mathbf{\Theta}\xspace}
\newcommand{\FM}{\mathbf{F}\xspace}
\newcommand{\xv}{\mathbf{x}\xspace}
\newcommand{\pv}{\mathbf{p}\xspace}

\newcommand{\zv}{\mathbf{z}\xspace}

\newcommand{\ve}{\mathbf{v}\xspace}

\newcommand{\avgd}{\overline{d}\xspace}
\newcommand{\avge}{\overline{\delta}\xspace}

\newcommand{\AHE}{AHNEE\xspace}

\newcommand{\nprox}{HMS-N\xspace}
\newcommand{\heprox}{HMS-E\xspace}
\newcommand{\aug}{\texttt{ExtendHG}\xspace}

\newcommand{\pts}{\texttt{PTS}\xspace}

\newcommand{\eigsh}{\texttt{Lanczos}\xspace}
\newcommand{\svds}{\texttt{TruncatedSVD}\xspace}

\newcommand{\pane}{\texttt{PANE}\xspace}
\newcommand{\aneci}{\texttt{AnECI}\xspace}
\newcommand{\conn}{\texttt{CONN}\xspace}
\newcommand{\netmf}{\texttt{NetMF}\xspace}
\newcommand{\lightne}{\texttt{LightNE}\xspace}

\newcommand{\hypertovec}{\texttt{Hyper2vec}\xspace}
\newcommand{\hypeboy}{\texttt{HypeBoy}\xspace}
\newcommand{\tricl}{\texttt{TriCL}\xspace}
\newcommand{\villain}{\texttt{VilLain}\xspace}
\newcommand{\anchorgnn}{\texttt{AnchorGNN}\xspace}
\newcommand{\gebe}{$\texttt{GEBE}^{\texttt{p}}$\xspace}
\newcommand{\biane}{\texttt{BiANE}\xspace}
\newcommand{\oursnopts}{\texttt{SHAPE-nopts}\xspace}
\newcommand{\oursbase}{\texttt{Base}\xspace}
\newcommand{\ours}{\texttt{SAHE}\xspace}
\newcommand{\avgnode}{\texttt{AVG\_nodeemb}\xspace}
\newcommand{\first}{\cellcolor{gray!45}}
\newcommand{\second}{\cellcolor{gray!30}} 
\newcommand{\third}{\cellcolor{gray!15}}

\newcommand{\stitle}[1]{\vspace{0.6mm}\noindent{\bf #1.}}
\newcommand{\eat}[1]{}
\newcommand{\report}[1]{}
\newcommand{\submission}[1]{#1}
\newcommand{\hardcode}[1]{#1}
\newcommand{\newrow}{\rowstyle{\color{black}}}

\newcommand{\revision}[1]{{{#1}}} 
\newcolumntype{B}{>{\color{black}}c} 

\newcommand*{\rowstyle}[1]{
  \gdef\@rowstyle{#1}%
  \@rowstyle\ignorespaces%
} 

\newcolumntype{=}{
  >{\gdef\@rowstyle{}}%
}

\newcolumntype{+}{
  >{\@rowstyle}%
}

\begin{document}
\title{Effective and Efficient Attributed  Hypergraph Embedding on Nodes and Hyperedges}

\author{Yiran Li, Gongyao Guo, Chen Feng, Jieming Shi}\authornote{Corresponding Author.}

\affiliation{%
  \institution{The Hong Kong Polytechnic University}
  \city{Hong Kong SAR}
  \state{China}
}
\email{{yi-ran.li, gongyao.guo}@connect.polyu.hk;  {chenfeng, jieming.shi}@polyu.edu.hk}


\begin{abstract}
An attributed hypergraph comprises nodes with attributes and hyperedges that connect varying numbers of nodes. \textit{Attributed hypergraph node and hyperedge embedding} (\AHE) maps nodes and hyperedges to compact vectors for use in important tasks such as node classification, hyperedge link prediction, and hyperedge classification. Generating high-quality embeddings is challenging due to the complexity of attributed hypergraphs and the need to embed both nodes and hyperedges, especially in large-scale data. Existing solutions often fall short by focusing only on nodes or lacking native support for attributed hypergraphs, leading to inferior quality, and struggle with  scalability on  large attributed hypergraphs.

We propose \ours, an efficient and effective approach that unifies node and hyperedge embeddings for \AHE computation, advancing the state of the art via comprehensive embedding formulations and algorithmic designs.  
First, we introduce two higher-order similarity measures, \nprox and \heprox, to capture similarities between node pairs and hyperedge pairs, respectively. These measures consider multi-hop connections and global topology within an extended hypergraph that incorporates  attribute-based hyperedges. \ours formulates the \AHE objective to jointly preserve all-pair \nprox and \heprox similarities. Direct optimization is computationally expensive, so we analyze and unify core approximations of all-pair \nprox and \heprox to solve them simultaneously. To enhance efficiency, we design several non-trivial optimizations that avoid iteratively materializing large dense matrices while maintaining high-quality results. 
Extensive experiments on diverse attributed hypergraphs and 3 downstream tasks, compared against 11 baselines, show that \ours consistently outperforms existing methods in embedding quality and is up to orders of magnitude faster.

\end{abstract}

\maketitle

\pagestyle{\vldbpagestyle}
\begingroup\small\noindent\raggedright\textbf{PVLDB Reference Format:}\\
\vldbauthors. \vldbtitle. PVLDB, \vldbvolume(\vldbissue): \vldbpages, \vldbyear.\\
\href{https://doi.org/\vldbdoi}{doi:\vldbdoi}
\endgroup
\begingroup
\renewcommand\thefootnote{}\footnote{\noindent
This work is licensed under the Creative Commons BY-NC-ND 4.0 International License. Visit \url{https://creativecommons.org/licenses/by-nc-nd/4.0/} to view a copy of this license. For any use beyond those covered by this license, obtain permission by emailing \href{mailto:info@vldb.org}{info@vldb.org}. Copyright is held by the owner/author(s). Publication rights licensed to the VLDB Endowment. \\
\raggedright Proceedings of the VLDB Endowment, Vol. \vldbvolume, No. \vldbissue\ %
ISSN 2150-8097. \\
\href{https://doi.org/\vldbdoi}{doi:\vldbdoi} \\
}\addtocounter{footnote}{-1}\endgroup

\ifdefempty{\vldbavailabilityurl}{}{
\vspace{.3cm}
\begingroup\small\noindent\raggedright\textbf{PVLDB Artifact Availability:}\\
The source code, data, and/or other artifacts have been made available at \url{https://github.com/CyanideCentral/AHNEE}.
\endgroup
}

\section{Introduction}\label{sec:intro}
An \textit{attributed hypergraph} captures higher-order relationships among a variable number of nodes through \textit{hyperedges}, and the nodes are  often associated with \textit{attribute} information. A hyperedge is a generalized edge that connects more than two nodes.
The unique characteristics of attributed hypergraphs have played important roles in various domains by describing the higher-order relationships among entities, e.g.,   social networks~\cite{bensonSimplicialClosureHigherorder2018}, genomic expression~\cite{fengHypergraphModelsBiological2021}, and online shopping sessions~\cite{hanSearchBehaviorPrediction2023}.
For example, a group-purchase activity of an item links a group of users together, naturally captured via a hyperedge, and the users carry their own profile attributes.

Network embedding is a fundamental problem in graph analytics, garnering attention from both academia~\cite{zhouNetworkRepresentationLearning2022} and industry~\cite{xuUnderstandingGraphEmbedding2021}, and has been studied on various types of simple graphs with  pairwise edge connections, such as homogeneous graphs~\cite{yangHomogeneousNetworkEmbedding2020, qiuLightNELightweightGraph2021, tsitsulinFREDEAnytimeGraph2021a} and attributed graphs~\cite{yangScalingAttributedNetwork2020, wuAttributedNetworkEmbedding2024}. However, network embedding on attributed hypergraphs is still in its early stages, with few native solutions that are efficient and effective in the literature. 

Hence, in this work, we focus on  the problem of \textit{Attributed Hypergraph Node and hyperEdge Embedding} (\AHE).
Given an attributed hypergraph with 
$n$ nodes and 
$m$ hyperedges, \AHE aims to generate compact embedding vectors for each node and hyperedge. Intuitively, node embeddings capture the hyperedge-featured topological and attribute information surrounding nodes, while hyperedge embeddings inherently capture the connections and attribute semantics of groups of nodes around hyperedges. The embeddings are valuable for downstream tasks: node embeddings facilitate node classification~\cite{huangHyper2vecBiasedRandom2019} and hyperedge link prediction~\cite{yuModelingMultiwayRelations2018a}, while hyperedge embeddings support hyperedge classification~\cite{yanHypergraphJointRepresentation2024a}. 

\eat{The embeddings are particularly useful for various downstream tasks: node embeddings are used for node classification~\cite{huangHyper2vecBiasedRandom2019}, and hyperedge link prediction~\cite{yuModelingMultiwayRelations2018a}; hyperedge embeddings are used for hyperedge classification~\cite{yanHypergraphJointRepresentation2024a}. These tasks are crucial for supporting real-world applications such as recommendation systems~\cite{xiaSelfSupervisedHypergraphTransformer2022}, crime forecasting~\cite{liSpatialTemporalHypergraphSelfSupervised2022}, and trust prediction~\cite{xuAdaptiveHypergraphNetwork2024}.}

It is highly challenging to design  native \AHE solutions, due to the complexity of attributed hypergraphs beyond simple graphs, and the need to simultaneously embed nodes and hyperedges, particularly for large attributed hypergraphs.
Effective node and hyperedge embeddings should capture both local and long-range information via multi-hop paths formed by hyperedges and nodes.
Besides, simply aggregating all node embeddings within a hyperedge to get its hyperedge embedding often results in suboptimal performance.
Incorporating these considerations into \AHE computation requires careful designs to ensure effectiveness and targeted optimizations for efficiency in large attributed hypergraphs.


Existing methods either do not natively support attributed hypergraphs or fail to perform efficiently on massive data.
As reviewed in Section \ref{sec:relatedwork}, an early study~\cite{zhouLearningHypergraphsClustering2007} uses the hypergraph Laplacian spectrum for node embeddings, while~\cite{huangHyper2vecBiasedRandom2019, huangNonuniformHyperNetworkEmbedding2020}  extend   graph-based node embedding to hypergraphs. However, these approaches typically do not consider attribute information or  hyperedge embedding generation, with some, like~\cite{zhouLearningHypergraphsClustering2007}, overlooking long-range connectivity.
A recent class of studies~\cite{wei_augmentations_2022, leeImMeWere2023, kimHypeBoyGenerativeSelfSupervised2023}  has developed hypergraph neural networks, which often incur significant computational overhead when applied to large-scale hypergraph data.
Another way is  converting attributed hypergraphs into bipartite or attributed graphs using star-expansion or clique-expansion, followed by applying graph embedding methods~\cite{yangScalingAttributedNetwork2020,wuBillionScaleBipartiteGraph2023}. However, the expansions dilute the representation of higher-order connections in hyperedges and result in dense graphs with high computational costs.

To tackle the challenges,  we propose \ours, a \underline{S}calable \underline{A}ttributed \underline{H}ypergraph node and hyperedge \underline{E}mbedding method that unifies the generation of node embeddings and hypergraph embeddings with high result quality and efficiency, advancing the state of the art for the problem of \AHE. 
We accomplish this via  comprehensive problem formulations and innovative algorithm designs.

We begin by considering an attribute-extended hypergraph $\HGA$, which integrates node attributes by constructing attribute-based hyperedges with appropriate weights, alongside the original hyperedges from the input attributed hypergraph. Importantly, on $\HGA$, we propose two measures: hypergraph multi-hop node similarity (\nprox) and hypergraph multi-hop hyperedge similarity (\heprox). \nprox captures higher-order connections and global topology between nodes by considering both original and attribute-based hyperedges in $\HGA$.   \heprox quantifies hyperedge similarities, but on a dual hypergraph of $\HGA$, where hyperedges are treated as nodes to preserve their multi-hop connections and global features. 
We then formulate the \AHE task as an optimization problem with the objective to approximate all-node-pair \nprox and all-hyperedge-pair \heprox matrices simultaneously.
Directly achieving this objective can be effective but computationally expensive, with time quadratic in the number of nodes and hyperedges.
To boost efficiency,  \ours unifies the approximations of   \nprox and \heprox matrices by identifying their shared core computations via theoretical analysis. Despite this unification, the process remains costly for calculation and materialization. To further reduce computational overhead, we  develop several optimization techniques that eliminate the need to iteratively materialize large dense matrices, enabling efficient approximation of high-quality node and hyperedge embeddings with guarantees.
We conduct extensive experiments on 10 real datasets, comparing \ours against 11 competitors over 3 tasks. The results show that \ours efficiently generates high-utility node and hyperedge embeddings, achieving superior predictive performance in node classification, hyperedge link prediction, and hyperedge classification tasks, while being up to orders of magnitude faster.

In summary, we make the following contributions in the paper.
\vspace{-\topsep}
\begin{itemize}[leftmargin=*]
    \item We build an attribute-extended hypergraph to incorporate  attribute information into hypergraph structures seamlessly, with a careful design to balance both aspects.
    \item  We design two  similarity measures \nprox and \heprox capturing higher-order connections and global topology of node pairs and hyperedge pairs,  respectively. The \AHE objective is formulated to  preserve all-pair \nprox and \heprox matrices.
    \item We develop several techniques to efficiently optimize the objective, including unifying the shared computations of node and hyperedge embeddings, accurate approximation of the similarity matrices, and avoiding iterative dense matrix materialization.
    \item Extensive experiments on diverse real datasets and 3 downstream tasks demonstrate  the effectiveness and efficiency of our method.

\eat{    \jieming{need one more contribution...} \yiran{Attribute-based augmentation with attribute-structure balance} \jieming{seems ok but not strong. For each item, write 2-3lines at most. Be consistent with above. } We incorporate the attribute information into the structure seamlessly, where the node and hyperedge weights are carefully designed to balance the attribute and structural information.
    \item In the attribute-extended hypergraph, two similarity measures (\ie \nprox and \heprox) are proposed to capture both multi-hop connectivities and global significance for nodes and hyperedges.
    \item  We propose the \ours algorithm to for the node and hyperedge embedding. To derive their embeddings simultaneously and efficiently, we identify their shared core computations and approximate the computations with a bounded error.
    \item Extensive experiments demonstrate that our \ours algorithm surpasses the 11 baselines, achieving superior overall embedding quality and significantly enhanced scalability.}
\end{itemize}

\begin{table}[!t]
\centering
\renewcommand{\arraystretch}{1.1}
\begin{footnotesize}
\caption{Frequently used notations.}\vspace{-3mm} \label{tbl:notations}
\resizebox{\columnwidth}{!}{
	\begin{tabular}{|p{0.70in}|p{2.5in}|}
		\hline
		{\bf Notation} &  {\bf Description}\\
		\hline
        $\ahg=\{\VS, \ES, \XM\}$ & An attributed hypergraph $\ahg$ with node set $\VS$, hyperedge set $\ES$, and node attribute matrix $\XM \in \mathbb{R}^{n\times q}$.\\\hline
        $n$, $m$ & The cardinality of $|\VS| = n $, and the cardinality of $|\ES| = m $. \\\hline
        $d(v), \delta(e)$ & The generalized degree of a node $v$ and a hyperedge $e$.  \\\hline
        $\HGA=\{\VS,\ESA\}$ & The extended hypergraph $\HG $ with hyperedge set $\ESA $ that incorporates $\ES$ and attribute-based hyperedges $\ES_K$.\\\hline
        $\vol(\HGA)$ & The volume of the hypergraph $\HGA$.\\\hline
        $\HM$ & The weighted incidence matrix of $\HGA$.\\\hline
        $\DM_v$, $\DM_e$, $\WM$ & The diagonal node degree matrix, hyperedge degree matrix, and hyperedge weight matrix of $\HGA$, respectively.\\\hline
        $\HGA'=\{\VSA', \ES'\}$ & The dual hypergraph of $\HGA$, where nodes in $\VSA'$ represent hyperedges in $\ESA$, and hyperedges in $\ES'$ represent nodes in $\VS$.\\\hline
        $p(v_i,v_j)$ & The random walk transition probability from $v_i$ to $v_j$.\\\hline
        $p_s(v)$ & The random walk stationary probability of node $v$.\\\hline
        $\PM$, $\PM'$ & The transition probability matrices of hypergraphs $\HGA$ and $\HGA'$.\\\hline
        $\pi(v_i,v_j)$ & The probability from $v_i$ to $v_j$ over infinite steps.\\\hline
        $\PiM^{(t)}$,$\PiM^{\prime (t)}$ & The t-step RWR matrix for $\HGA$ and $\HGA'$, respectively. \\\hline
        $\tlog(\cdot), \tlog^\circ(\cdot)$. & $\tlog(x)=\log(\max\{x,1\})$, $\tlog^\circ(\cdot)$ is element-wise $\tlog(\cdot)$.\\\hline
        $\psi(v_i,v_j)$ & \nprox similarity between nodes $v_i$ and $v_j$ of $\HGA$.\\\hline
        $\psi'(e_i,e_j)$ & \heprox similarity between hyperedges $e_i$ and $e_j$ of $\HGA$.\\\hline
        $\PsM$, $\PsM'$ & Similarity matrices for \nprox and \heprox, respectively.\\\hline
        $\ZM_\VS$, $\ZM_\ES$ & The $n\times k$ node embedding matrix and $m\times k$ hyperedge embedding matrix, respectively.\\\hline
	\end{tabular}
}
\end{footnotesize}
\vspace{-2mm}
\end{table}
\section{Preliminaries} \label{sec:prelim}
An attributed hypergraph is denoted as $\ahg=\{\VS, \ES, \XM\}$, where $\VS$ is a set of $n$ nodes, $\ES$ is a set of $m$ hyperedges, and $\XM\in \mathbb{R}^{n\times q}$ is the node attribute matrix.
Each node $v$ in $\VS$ has a $q$-dimensional attribute vector given by the $i$-th row of $\XM$.
A hyperedge $e\in \ES$ is a subset of $\VS$ containing at least two nodes, and a node $v$ is incident to $e$ if $v\in e$.
Figure \ref{fig:attribute_hypergraph} gives an  attributed hypergraph $\ahg$ with six nodes and three hyperedges (e.g., $e_2=\{v_3,v_4,v_5\}$), where each node is associated with an attribute vector. Let $\gamma(v, e)$ be the hyperedge-dependent weight of node $v$ in hyperedge $e$, defaulting to 1 if $v\in e$ and unweighted, or 0 if $v \notin e$.
Each hyperedge $e\in \ES$ carries a weight $w(e)$, defaulting to 1.  The generalized degree of a node $v\in\VS$ is  $d(v)=\sum_{e\in \ES} w(e)\gamma(v, e)$, summing the weighted contributions of $v$ across incident hyperedges. The generalized degree of a hyperedge $e
\in\ES$ is $\delta(e)=\sum_{v\in e}\gamma(v,e)$, aggregating the weight of nodes within $e$. The {volume} of $\ahg$,  $\vol(\ahg)=\sum_{v\in\VS}d(v)$, measures its total weighted connectivity by summing the generalized node degrees.
Let $\HMO\in \mathbb{R}^{m\times n}$ be the incidence matrix of $\ahg$, where each entry $\HMO[i,j]=\gamma(v_j, e_i)$ if $v_j\in e_i$, otherwise $\HMO[i,j]=0$. 




\eat{Let diagonal matrices $\DM_v\in \mathbb{R}^{n\times n}$ and $\DM_e\in \mathbb{R}^{m\times m}$ represent the generalized node degree matrix and hyperedge degree matrix of $\ahg$, where the diagonal entry is $\DM_v[j,j]=d(v_j)$ for $v_j\in \VS$ and $\DM_e[i,i]=\delta(e_i)$ for $e_i\in \ES$, respectively. The diagonal hyperedge weight matrix is denoted by $\WM$ such that $\WM[i,i]=w(e_i)$ for $e_i\in \ES$.}

\stitle{Node and Hyperedge Embeddings}
For the input $\ahg$, \AHE aims to compute an $n\times k$ embedding matrix $\ZM_\VS$ where each row $\ZM_\VS[i]$ is the embedding vector for node $v_i\in \VS$ (\textit{node embedding}), and also an $m\times k$ embedding matrix $\ZM_\ES$ where each row $\ZM_\ES[j]$ is the embedding vector for hyperedge $e_j\in \ES$ (\textit{hyperedge embedding}).

\stitle{Tasks} 
In the attributed hypergraph $\ahg$, we focus on three 
significant tasks: \textit{node classification} and \textit{hyperedge link prediction} with
node embeddings; \textit{hyperedge classificaiton} with hyperedge embeddings. 
\vspace{-\topsep}
\begin{itemize}[leftmargin=*]
    \item \textit{Node Classification.} For node $v_i$, the goal is to predict its class label by feeding its node embedding $\ZM_\VS[i]$ into a trained classifier.
    \item \textit{Hyperedge Link Prediction.}
    Given nodes $\{v_i,v_j,...,v_k\} \subset\VS$, the task is to use their node embeddings $\{\ZM_\VS[i],\ZM_\VS[j],...,\ZM_\VS[k]\}$ to predict whether these nodes form a hyperedge or not. 
    \item \textit{Hyperedge Classification.} For a hyperedge $e_i$, the goal is to use the hyperedge embedding $\ZM_\ES[i]$ to predict its class label. 
\end{itemize}

Table \ref{tbl:notations} lists the frequently used notations in our paper.

\section{Similarities and Objectives}\label{sec:problemformulation}
As explained in Section \ref{sec:intro}, effective node and hyperedge embeddings in \AHE should capture the closeness among nodes and collective affinities among hyperedges, highlighting higher-order structures and attribute influences.
To this end, we first extend the attributed hypergraph $\ahg$ by adding attribute-based hyperedges, forming an attribute-extended hypergraph $\HGA$, in which, appropriate hyperedge-dependent node weights and hyperedge weights are assigned to balance structural and attribute information in Section \ref{sec:aughg}.
Then in Section \ref{sec:nodeSim}, 
we introduce hypergraph multi-hop node similarity (\nprox) on $\HGA$ to quantify node similarity. This considers multi-hop node and hyperedge connections and node significance in $\HGA$, measured by a random walk model over its topology. Section \ref{sec:edgeSim} designs hypergraph multi-hop hyperedge similarity (\heprox),  formulated similarly but on the dual hypergraph of $\HGA$.
Both \nprox and \heprox are symmetric for evaluating pairwise relationships. Our \AHE objective is to approximate all $n\times n$ node-pair \nprox and all $m\times m$ hyperedge-pair \heprox similarities. Section \ref{sec:basemethod} presents a preliminary method to directly solve this problem. 

\begin{figure}[!t]
    \includegraphics[width=0.67\columnwidth]{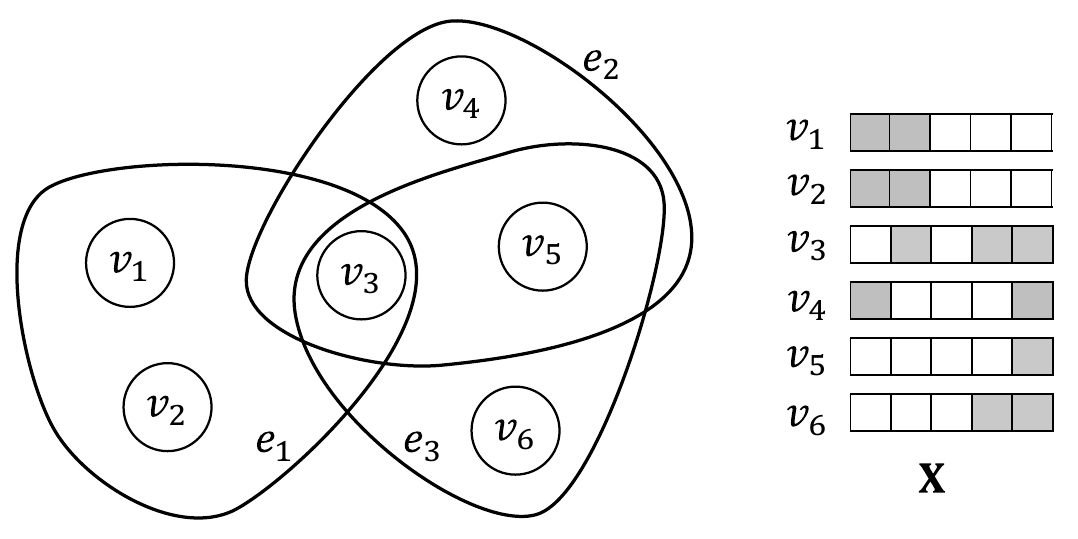}
    \vspace{-3mm}
    \caption{An example of attributed hypergraph $\ahg$.} \label{fig:attribute_hypergraph}
\vspace{-4mm}
\end{figure}

\eat{
As explained in Section \ref{sec:intro}, effective node and hyperedge embeddings in \AHE should maintain relational closeness among nodes and collective affinities among hyperedges, highlighting higher-order structures and attribute influences. 

To this end, we first 
extend the structure of the input hypergraph with the attribute information in $\XM$ to get an attribute-extended hypergraph $\HGA$ in Section \ref{sec:aughg}. Then, we introduce the hypergraph multi-hop similarity for nodes and hyperedges, \nprox and \heprox,  to quantify the pairwise node similarity and hyperedge similarity in. On this basis, we formalize the objective of our \AHE: approximating all $n\times n$ node-pair similarities and all $m\times m$ hyperedge-pair similarities with the embeddings.
Finally, a base method is proposed to achieve the objective directly. 
}

\eat{For node embeddings, we seek to capture both the structural similarity and the attribute similarity, encompassing local neighborhoods and long-range influences. This poses challenges, as it requires modeling higher-order connections beyond direct adjacencies while accounting for global topology. For hyperedge embeddings, simply aggregating node embeddings cannot fully reflect a hyperedge’s collective identity, as our experiments confirm. Despite the diverse sizes and roles of hyperedges, we seek to preserve their inherent identities and higher-order relationships within the network.}

\subsection{Attribute-Extended Hypergraph} \label{sec:aughg}


The literature on attributed data~\cite{liEfficientEffectiveAttributed2023,fengGraphRepresentationAttributed2025}  shows that it is effective for downstream task performance to combine attribute information with network topology, by considering each node's K-nearest neighbors defined by attribute similarity. We adopt this approach to construct an extended hypergraph $\HGA$ from the input $\ahg$, with dedicated designs tailored for hypergraph structures while balancing topological and attribute information. 
Specifically, for a node $v_i$, we first get its local neighbor set $\knn(v_i)$, comprising the top-$K$ most similar nodes $v_j$ ranked by cosine similarity $\cs(v_i,v_j)$ of their attributes.  
\revision{
Then we define an \textit{attribute-based hyperedge} as $\knn(v_i)\cup\{v_i\}$ with $K+1$ nodes. 
Intuitively, it connects nodes with similar attributes into a hyperedge.}


\revision{\stitle{Example} 
Given the hypergraph $\ahg$ in Figure \ref{fig:attribute_hypergraph}, we construct five attribute-based hyperedges $\{e_4,\dots,e_8\}$ for nodes $v_1,\dots,v_5$, respectively. With $K=2$, Figure \ref{fig:augmented-hypergraph3} illustrates two examples: $e_4$, formed by $v_1$ and its two most similar nodes, and $e_8$, formed by $v_5$ and its two most similar nodes. Each hyperedge represents a local neighborhood of nodes with high attribute similarity.
}

    

\begin{figure}[!t]
    \includegraphics[width=0.9\columnwidth]{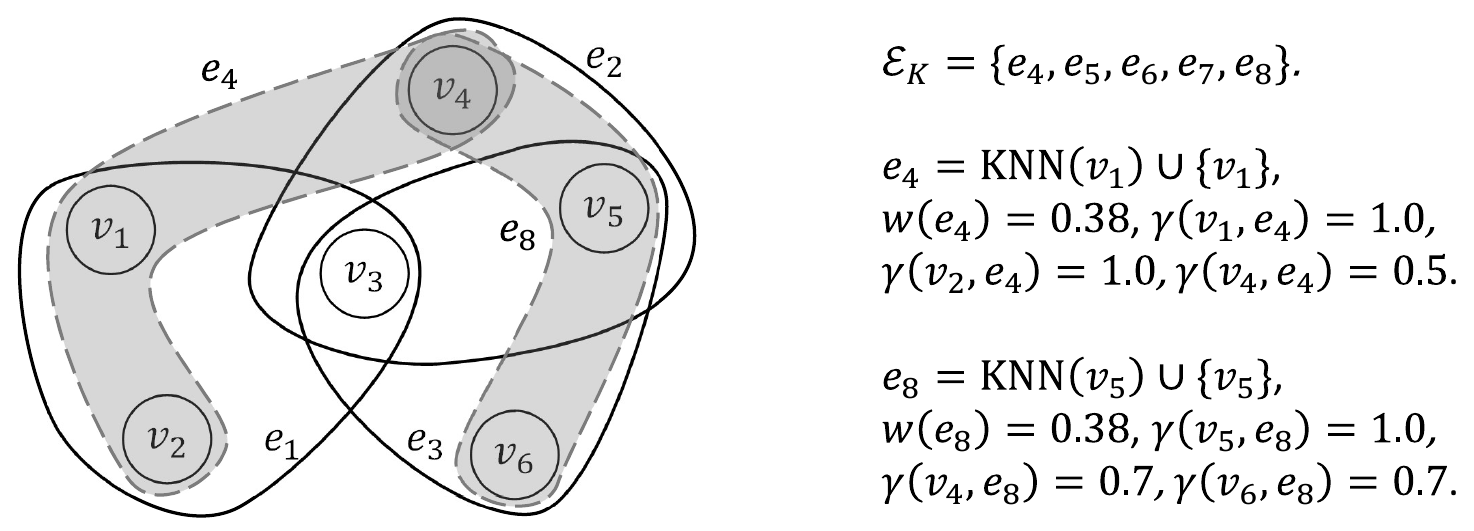}
    \vspace{-2.5mm}
    \caption{Extended hypergraph $\HGA$.}
    \label{fig:augmented-hypergraph3}
    \vspace{-3mm}
\end{figure}


For all $n$ nodes in $\ahg$, we create a set $\ESK$ of  $n$ attribute-based hyperedges, leading to the extended hypergraph $\HGA=\{\VS,\ESA\}$ with hyperedge set $\ESA=\ES\cup\ESK$ of size $m+n$. Unlike nodes in an original hyperedge, the nodes in an attribute-based hyperedge $e$ built from $\knn(v_i)$ require varying weights based on their attribute similarities to $v_i$.
Therefore, we assign \textit{hyperedge-dependent node weights }$\gamma(v_j,e)$ to node $v_j$ in the hyperedge $e$, using its attribute similarity to $v_i$: $\gamma(v_j,e)= \cs(v_i,v_j), \text{ for } v_j\in e \text{ and }  e=\knn(v_i)\cup\{v_i\}$.

$\HGA$ includes 
$m$ hyperedges in $\ES$  and $n$ attribute-based hyperedges in $\ESK$, representing structural and attribute information, respectively. The values of $n$ and $m$ can vary significantly across datasets. For example, the Amazon dataset has about $n=2.27$ million nodes and $m=4.28$ million hyperedges, while MAG-PM has  $n=2.35$ million nodes and $m=1.08$ million hyperedges. A large $n$ may cause attribute-based hyperedges to overshadow the original topology, and vice versa.
This disparity
extends to their volumes $\vol(\ES)=\sum_{e\in\ES}w(e)\delta(e)$, which can influence similarity measures by skewing transitions to $\ES$ or $\ESK$, affecting embedding priorities. Thus, we balance structural and attribute aspects in $\HGA$.  
We achieve this by adjusting hyperedge weights in $\ESK$, balancing the volumes of $\ES$ and $\ESK$.
For $\ES$, the volume is $\textstyle\vol\left(\ES\right)=\sum_{e\in \ES}|e|$, when $\gamma(v,e)$ and $w(e)$ are with unit weights.
For  $\ESK$, the volume is $\vol(\ESK)=w(e)\sum_{e\in\ES_K}\sum_{v\in e}\gamma(v,e)$, where hyperedge-dependent node weight $\gamma(v,e)$ is assigned ahead.
We enforce: 
\begin{equation} \textstyle
\textstyle\beta\vol\left(\ES\right)=\vol\left(\ESK\right),    
\end{equation}
where parameter $\beta$ controls the balance on structure versus attributes, and the default value is 1. Then we get a uniform weight  $w(e)$ of each attribute-based hyperedge $e\in\ESK$ by
\begin{equation} \label{eq:knn-edge-weight} \textstyle
    \textstyle w(e)=\beta\vol(\ES)\Big/\sum_{e\in\ES_K}\sum_{v\in e}\gamma(v,e), \forall  e\in\ESK. 
\end{equation}

The extended hypergraph $\HGA=\{\VS,\ESA\}$  has an $(m+n)\times n$ weighted incidence matrix $\HM$, where $\HM[i,j]=\gamma(v_j, e_i)$.
The first $m$ rows of $\HM$ are from the incidence matrix $\HM_0$, and the last $n$ rows are from the attribute-based hyperedges, forming a submatrix $\HM_K$. Then, the matrix $\HM$ can be written as $\HM=\left[\begin{smallmatrix}\HM_0\\\HM_K\end{smallmatrix}\right]$, where columns corresponding to the nodes in $\VS$.  For $\HGA$, let $\WM\in \mathbb{R}^{(m+n)\times (m+n)}$ be the diagonal matrix of hyperedge weights. 
In addition, $\DM_v\in \mathbb{R}^{n\times n}$ is the node degree matrix, with $\DM_v[i,i]=d(v_i)$ representing  the generalized degree of node $v_i$ in $\HGA$, and $\DM_e\in \mathbb{R}^{(m+n)\times (m+n)}$ is the hyperedge degree matrix with $\DM_e[i,i]=\delta(e_i)$ representing the generalized degree of hyperedge $e_i$ in $\HGA$.

The above idea of constructing $\HGA$ from $\ahg=\{\VS,\ES,\XM\}$ is summarized in Algorithm \ref{alg:augment}. Lines 1-4 generate $n$ attribute-based hyperedges $e_i=\knn(v_i)\cup \{v_i\}$, each capturing a node’s $K$-nearest neighbors based on attribute similarity from $\XM$ (Line 2).  Lines 3-4 assign hyperedge-dependent node weights  $\gamma(v_j, e_i)$ in  $\HM_K$.  Line 5 constructs the hyperedge weight matrix $\WM$ as a diagonal matrix from concatenated weight vectors (denoted by $\|$). The $m$ hyperedges in $\ES$ have weights of 1 via $\mathbf{1}_m$, a length-$m$ vector of ones, while the $n$ attribute-based hyperedges receive a weight per Eq. \eqref{eq:knn-edge-weight}, with $\mathbf{1}_m\HMO\mathbf{1}_n$ representing $\vol(\ES)$.
Line 6 builds the $(m+n)\times n$ weighted incidence matrix $\HM$ by stacking $\HMO$ and $\HM_K$, and computes the diagonal degree matrices $\DM_v$ and $\DM_e$ of $\HGA$.  The algorithm returns $\HM$, $\DM_v$, $\DM_e$, and $\WM$, representing $\HGA$. 
Lines 1-4 require $O(n\log n+nqK)$ time, leveraging efficient KNN queries (\eg~\cite{douze2024faiss}), where $q$ is the attribute dimension. Lines 5-6 operate in $O(nK+n\avgd)$ time and space, proportional to $\HM$’s nonzero entries, with $\avgd$ as the average hyperedge incidences per node in $\ahg$. Thus, Algorithm \ref{alg:augment} achieves log-linear time complexity and linear space complexity for generating $\HGA$.

\begin{algorithm}[!t]
\caption{\texttt{ExtendHG}} \label{alg:augment} 
\small
\KwIn{Attributed hypergraph $\ahg=\{\VS, \ES, \XM\}$, parameter $K$
}

\For{each $v_i\in\VS$}{ 
    Attribute-based hyperedge $e_i\gets \knn(v_i)\cup\{v_i\}$\;
   \For{each $v_j\in e_i$}{ 
    $\HMK[i,j]\gets\gamma(v_j, e_i)$\;
    }
}
$\WM\gets\diag\left(\mathbf{1}_m \Vert \frac{\mathbf{1}_m\transpose\HM_0\mathbf{1}_n}{\mathbf{1}_n\transpose\HM_K\mathbf{1}_n}\mathbf{1}_n\right)$\;
  $\HM\gets\left[\begin{smallmatrix}\HM_0\\\HM_K\end{smallmatrix}\right],\ \DM_v\gets\diag\left(\HM\transpose\WM\mathbf{1}_{m+n}\right),\ \DM_e\gets\diag\left(\HM\mathbf{1}_{n}\right)$\;
\Return Extended hypergraph $\HGA$ (i.e., $\HM$, $\DM_v$, $\DM_e$, $\WM$)\; 
\end{algorithm}

\subsection{Hypergraph Multi-Hop Node Similarity: \nprox}\label{sec:nodeSim}

Intuitively, the resultant node embeddings should capture the complex relationships between nodes, preserving both structural and attribute similarities across the extended hypergraph $\HGA$.
This is challenging due to the need to model higher-order connections in hyperedges while integrating the hypergraph's global topology. Common similarity measures like Personalized PageRank ~\cite{yangEfficientAlgorithmsPersonalized2024} effectively capture node significance but are limited to pairwise interactions, not applicable to hypergraphs. Other hypergraph definitions~\cite{takaiHypergraphClusteringBased2020a,ameranisFastAlgorithmsHypergraph2024} consider submodular hyperedges, which are unnecessary for embeddings and incur expensive all-pair computations.

\revision{
    To address these challenges, we propose a hypergraph multi-hop similarity measure for nodes (\nprox) over $\HGA$.
    The key insights include  (i) capturing multi-hop connectivity between nodes and (ii) leveraging the global significance of nodes in $\HGA$, both relying on random walks adapted to the hypergraph structure.
}


\stitle{\nprox Formulation} We begin by defining the transition probability  $p(u,v)$   for random walks on $\HGA$, considering hyperedge sizes, weights $w(e)$, generalized node degrees $d(u)$, and, crucially, hyperedge-dependent node weights $\gamma(u,e)$, reflecting attribute similarities. 
$p(u,v)$ involves two hops: from node $u$ to a hyperedge and from the hyperedge to node $v$. First, unlike prior definitions~\cite{chitraRandomWalksHypergraphs2019a}, an incident hyperedge $e$ is selected with probability proportional to $w(e)\gamma(u,e)/d(u)$, where $\gamma(u,e)$ emphasizes attribute affinity, and $w(e)$   balances structural and attribute significance. Second, within the chosen hyperedge $e$, the node $v$ is selected with probability proportional to $\gamma(v,e)/\delta(e)$, prioritizing nodes with stronger attribute  ties.
Therefore, the transition probability is 
\begin{equation}\label{eq:weighted-hgrw} 
\textstyle
    p(u, v) = \sum_{e\in \ESA} \frac{w(e)\gamma(u, e)}{d(u)} \frac{\gamma(v, e)}{\delta(e)},
\end{equation}
and accordingly the $n\times n$ transition matrix $\PM$ with each entry $\PM[i,j]=p(v_i,v_j)$ can be written as
\begin{equation}\label{eq:hgrw-matrix} \textstyle
    \PM = \DM_v^{-1}\HM\transpose \WM \DM_e^{-1} \HM.
\end{equation}






First, to capture the multi-hop connectivity between nodes, we use the random walk with restart (RWR) model.
Specifically, in the extended hypergraph $\HGA$, at each step, the walk either teleports back to $u$ with probability $\alpha\in[0,1)$ or transitions to a node with probability $1-\alpha$, following the transition matrix $\PM$ in Equation \eqref{eq:hgrw-matrix}. Let $\pi(v_i, v_j)$ denote the limiting probability that a random walk starting from node $v_i$ reaches node $v_j$ after infinitely many iterations, reflecting $v_j$'s significance to $v_i$ across local and global levels.
The probability $\pi^{(t)}(v_i, v_j)$ of reaching any node $v_j$ from any node $v_i$ after $t$ steps is represented by the stochastic matrix $\PiM^{(t)} \in\RN^{n\times n}$ in Equation \eqref{eq:rwr-recursive},  where $\PiM^{(t)}[i,j]$ is $\pi^{(t)}(v_i, v_j)$. 
\begin{equation}\label{eq:rwr-recursive} \textstyle
    \PiM^{(0)}=\IM_n,\  
    \PiM^{(t+1)} = \alpha \IM_n + (1-\alpha) \PiM^{(t)}\PM,
\end{equation}
with $\IM_n$ as the $n\times n$ identity matrix. The non-recursive formula is 
\begin{equation}\label{eq:rwr-formula} \textstyle
    \PiM^{(t)} = \textstyle{\sum_{i=0}^{t-1} \alpha(1-\alpha)^{i} \PM^i + (1-\alpha)^{t}\PM^t}.
\end{equation} 

\eat{When $t$ goes to infinite steps, $\PiM^{(t)}$ converges to $\PiM$ as follows. 
\begin{equation}\label{eq:rwr-infity}
\PiM=\PiM^{(\infty)}= \sum\limits_{i=0}^\infty \alpha(1-\alpha)^i \PM^i,
\end{equation}}

Let $\PiM$ represent $\PiM^{(t=\infty)}$ with infinite steps to converge. As one may note, Equations \eqref{eq:rwr-recursive} and \eqref{eq:rwr-formula} have similar forms to those in simple graphs. However, our formulation is based on a derivation tailored to the extended hypergraph $\HGA$. Accordingly, the transition includes the selection of a hyperedge and a node therein. Moreover, we do not employ $\PiM$ as the similarity matrix, but instead take into account the significance of each specific node as below.



Second, note that in a connected and nontrivial $\HGA$, the transition matrix $\PM$ is irreducible and aperiodic, ensuring a unique stationary distribution $\pv_s=\pv_s\PM$~\cite{zhouLearningHypergraphsClustering2007}. Let $p_s(v)$ be the element in $\pv_s$ w.r.t. node $v$. Then, $p_s(v)$ is the probability that an arbitrary random walk ends at $v$, indicating the significance of $v$. Moreover, the significance $p_s(v)$ can be calculated by $p_s(v)=d(v)/\vol(\HGA)$.



Finally, combining the multi-hop connectivity and node significance, we define the \nprox similarity measure as 
\begin{equation}\label{eq:def_psi}  \textstyle
    \psi(u,v)=\tlog\frac{\pi(u,v)}{p_s(v)},
\end{equation}
where $\pi(u,v)$ is divided by $p_s(v)$ to offset the inherent global significance of $v$ (\ie its degree centrality), and thus isolate the specific relational strength between $u$ and $v$ for a balanced and embedding-friendly measure. Moreover, the truncated logarithm, $\tlog x=\log(\max(x,1))$, is applied to stabilize the ratio against small $p_s(v)$ in large hypergraphs.

By Lemma \ref{lm:symmetry} (\submission{proof in the technical report~\cite{report}}\report{proof in Appendix \ref{app:proofs}}), we establish the symmetry of \nprox, ensuring the mutual similarity between two nodes. This is critical for embedding, as the dot product of corresponding embedding vectors should preserve their mutual \nprox similarity.
\begin{lemma}\label{lm:symmetry}
    For any nodes $v_i,v_j\in\VS$, we have $\frac{\pi(v_i, v_j)}{p_s(v_j)}=\frac{\pi(v_j, v_i)}{p_s(v_i)}$.
\end{lemma}

With Eq. \eqref{eq:def_psi}, we can express the \nprox matrix for all $n\times n$ node pairs in $\HGA$ as
\begin{equation}\label{eq:similarity-matrix} \textstyle
    \PsM = \tlog\ewise\left(\vol(\HGA)\PiM\DM_v^{-1}\right),
\end{equation}
where $\PsM[i,j]=\psi(v_i,v_j)$ and $\tlog^\circ(\cdot)$ means the element-wise truncated logarithm. Also, $\PsM$ is a symmetric matrix, as the element-wise $\tlog\ewise(\cdot)$ function preserves the symmetry established in Lemma \ref{lm:symmetry}. 

\stitle{Node Embedding Objective}
We aim to use the dot product of two node embeddings to preserve the \nprox between nodes. Specifically, let $\ZM_\VS$ denote the $n\times k$ embedding matrix where each row is a node embedding. Then, the node embedding problem of solving $\ZM_\VS$ is formulated as follows, where $\|\cdot\|_F$ is the Frobenius norm.
\begin{equation}\label{eq:nodeObj}  \textstyle
    \ZM_\VS=\argmin_{\ZM\in\RN^{n\times k}}\|\PsM-\ZM\ZM\transpose\|_F^2.
\end{equation}

\begin{figure}[!t]
\centering
\includegraphics[width=0.8\columnwidth]{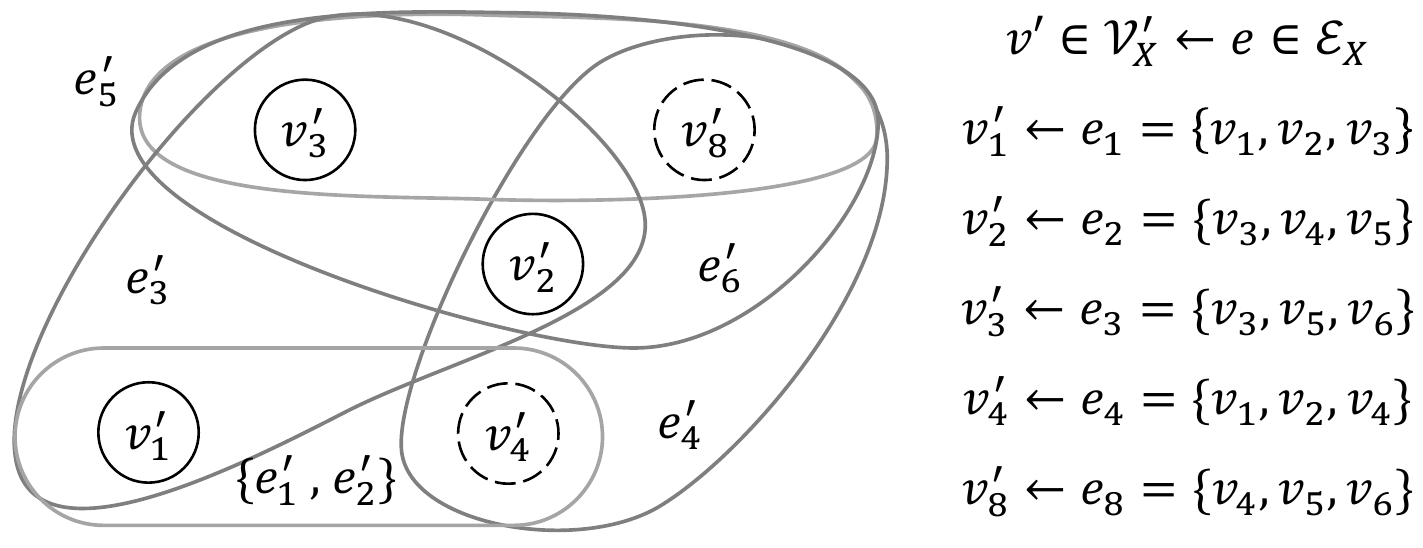}
\vspace{-3mm}
\caption{Dual hypergraph $\HGA'$.} \label{fig:dual_hypergraph}
\vspace{-4mm}
\end{figure}

\subsection{Hypergraph Multi-Hop Hyperedge Similarity: \heprox}\label{sec:edgeSim}

\eat{Given a hyperedge $e$ with a set of nodes, a simple aggregation of the embeddings of these nodes cannot adequately represent the hyperedge as a whole, since hyperedges in $\HGA$ are higher-order connections, and their similarity reflects shared structural and attribute-based contexts that node-centric methods may miss. 
Therefore, we design a native hyperedge-to-hyperedge multi-hop similarity measure, \heprox,   to enable effective hyperedge embeddings.}

\revision{
    The definition of \heprox aligns with the principles of \nprox, but it is derived using the dual hypergraph $\HGA'$ of $\HGA$, where nodes and hyperedges swap their roles.
     $\HGA'$ can be obtained by transposing the incidence matrix $\HM$.
}
In $\HGA'=\{\VSA',\ES'\}$, the new node set $\VSA'$ includes $m+n$ nodes and the new hyperedge set $\ES'$ has $n$ hyperedges. Here, each node $v_i'\in\VSA'$ represents the hyperedge $e_i\in\ESA$ of $\HGA$. Conversely, each hyperedge $e_j'\in\ES'$ contains nodes in $\VSA'$ which correspond to the hyperedges in $\HGA$ incident to $v_j\in\VS$. 

\revision{\stitle{Example} Figure \ref{fig:dual_hypergraph} illustrates the dual hypergraph $\HGA'$ derived from the original $\HGA$ in Figure \ref{fig:augmented-hypergraph3}.
In $\HGA'$, hyperedge $e_1$ of $\HGA$ is represented by node $v_1'$ and is connected to $v_2'$ and $v_3'$ through $e_3'$, which corresponds to node $v_3$ in $\HGA$.
This representation enables the analysis of similarity between $e_1$ and $e_2$ in $\HGA$ based on the relationships between $v_1'$ and $v_2'$ in $\HGA'$.}


The dual hypergraph $\HGA'$ has a weighted incidence matrix $\HM'=\HM\transpose\WM$, hyperedge weight $w(e')=1$ ($e'\in\ES'$), and $\WM'=\IM_n$. Hyperedge-dependent node weights are $\gamma'(v_i', e_j')=w(e_i)\gamma(v_j, e_i)$, retaining the influence of hyperedge weights in $\HGA$.
The generalized node degree, hyperedge degree, and hypergraph volume of $\HGA'$ are:
\begin{equation*} \textstyle
    d(v_i')=\delta(e_i)w(e_i), \delta(e_j')=d(v_j),\vol(\HGA')=\vol(\HGA),
\end{equation*} 
with matrix forms $\DM_v'=\WM\DM_e$ and $\DM_e'=\DM_v$.

\stitle{\heprox Formulation}  
In the dual hypergraph $\HGA'$, a random walk transitions from node $v_i'$ (hyperedge $e_i$ in $\HGA$) to node $v_j'$ (hyperedge $e_j$ in $\HGA$) via a hyperedge $e'$ (node $v$ shared by $e_i$ and $e_j$), with probability proportional to $\gamma'(v_i',e')\gamma'(v_j',e')/d(v_i')$, which is $w(e_j)\gamma(v,e_i)\gamma(v,e_j)/\delta(e_i)$ in $\HGA$. Aggregating over all shared nodes, the transition probability between $e_i$ and $e_j$ effectively captures the strength of their overlap. 
The transition probability between hyperedges $e_i, e_j$ in the original hypergraph $\HGA$ is 
\begin{equation} \textstyle
    p'(e_i,e_j)=p(v_i',v_j')=\sum_{v\in \VS} \frac{\gamma(v, e_i)}{\delta(e_i)} \frac{w(e_j)\gamma(v, e_j)}{d(v)},
\end{equation}
where the function $p'(\cdot,\cdot)$ with a prime indicates the transition between hyperedges in the original hypergraph $\HGA$. Moreover, the $(m+n)\times(m+n)$ transition matrix $\PM'$ can be expressed as
\begin{equation}\label{eq:hgrw-dual} \textstyle
    \PM'=(\DM_e\WM)^{-1}(\HM\transpose\WM)\transpose\DM_v^{-1}\HM\transpose\WM = \DM_e^{-1}\HM\DM_v^{-1}\HM\transpose\WM.
\end{equation}

Similar to \nprox, \heprox between hyperedges $e_i$ and $e_j$  (i.e., nodes $v_i'$ and $v_j'$ in $\HGA'$ ) also considers their multi-hop connectivity and  global significance within $\HGA'$.
Let $\PiM^{\prime (t)}$ be the probability of transitioning between hyperedges (represented as nodes in $\HGA'$) over $t$ steps.
In the limit, $\PiM^{\prime (\infty)}$ reflects the long-term likelihood of reaching one hyperedge from another. The matrix $\PiM^{\prime (t)}$ for the dual hypergraph is computed iteratively by substituting $\PM'$ into Eq. \eqref{eq:rwr-recursive}.
Let $\PiM^\prime$ denote $\PiM^{\prime(t=\infty)}$ with infinite steps to converge.
Then, the probability that a random walk from hyperedge $e_i$ reaches hyperedge $e_j$ in $\HGA$ is $\pi'(e_i,e_j)=\pi(v_i', v_j')=\PiM'[i,j]$. 

The global significance of hyperedge $e_j$ is the significance of its corresponding node $v_j'$ in $\HGA'$, $p_s(v'_j)$, calculated as :
\begin{equation} \textstyle
p_s'(e_j)=p_s(v_j')=d(v_j')/\vol(\HGA')=\delta(e_j)w(e_j)/\vol(\HGA).
\end{equation}
Finally, the \heprox of hyperedges $e_i$ and $e_j$, $\psi'(e_i, e_j)$, is 
\begin{equation} \textstyle
    \psi'(e_i, e_j)=\tlog\frac{\pi'(e_i,e_j)}{p_s'(e_j)}.
\end{equation}
The corresponding \heprox matrix for all hyperedge pairs is 
\begin{equation}\label{eq:similarity-matrix-dual} \textstyle
\PsM' = \tlog\ewise\left(\vol(\HGA)\PiM'\DM_e^{-1}\WM^{-1}\right).
\end{equation}

\eat{As the dual hypergraph $\HGA'$ represents the connections between hyperedges in $\ESA$, the similarity between nodes $v_i'$ and $v_j'$ indicates the similarity between hyperedges $e_i$ and $e_j$.
Therefore, the Hypergraph Multi-hop Similarity for hyperedges (\heprox) is defined by $\psi'(e_i, e_j)=\psi(v_i', v_j')=\tlog\frac{\pi(v_i',v_j')}{p_s(v_j')}$. \heprox is a symmetric similarity measure, which follows from the symmetry of node similarity $\psi(v_i',v_j')$ for the dual $\HGA'$. The similarity matrix for all hyperedges of $\ESA$ is expressed as:
\begin{equation}\label{eq:similarity-matrix-dual}
\PsM' = \tlog\ewise\left(\vol(\HGA)\PiM'\DM_e^{-1}\WM^{-1}\right).
\end{equation}}

\stitle{Hyperedge Embedding Objective}
We aim to use the dot product of two hyperedge embeddings to preserve the \heprox between the hyperedges. 
Let $\ZM_\ES$ denote the $m\times k$ embedding matrix where each row is the embedding vector for a hyperedge $e\in\ES$, and $\PsM'_\ES$ denote the $m\times m$ submatrix of $\PsM'$ that represents the \heprox similarity between hyperedges in $\ES$. Then, the hyperedge embedding problem of solving $\ZM_\ES$ can be formulated as
\begin{equation}\label{eq:edgeObj} \textstyle
\ZM_\ES=\argmin_{\ZM\in\RN^{m\times k}}\|\PsM'_\ES-\ZM\ZM\transpose\|_F^2.
\end{equation}

\subsection{A Base Method} \label{sec:basemethod}
 
In this section, we introduce a basic method to directly address the node and hyperedge embedding objectives in Equations \eqref{eq:nodeObj} and \eqref{eq:edgeObj}. The purpose of   this base method is two-fold. First, it verifies the effectiveness of the proposed measures \nprox and \heprox, by demonstrating superior quality over existing methods. Second, it establishes a foundation for our final method \ours, described in Section \ref{sec:finalmethod}, which achieves comparably high embedding quality with significantly improved efficiency.



For node embeddings, the main idea  is to factorize the \nprox matrix $\PsM$. Recall that  $\PsM$ in Eq. \eqref{eq:similarity-matrix} relies on $\PiM^{(t=\infty)}$ in Eq. \eqref{eq:rwr-recursive}, which is an infinite sum of powers of the transition matrix $\PM$. To be tractable, we approximate $\PiM$ by at most $t=T$ steps, resulting in $\PiM^{(T)}$ by Eq. \eqref{eq:rwr-formula}. 
Accordingly, we can derive the approximate \nprox matrix $\PsM_T$ by replacing $\PiM$ with $\PiM^{(T)}$ in Eq. \eqref{eq:similarity-matrix}. 
Now, the focus is to factorize the symmetric matrix $\PsM_T$ to get node embeddings $\ZM_\VS\in\RN^{n\times k}$.
Specifically, we utilize the eigendecomposition $\PsM_T= \QM \LdM \QM\transpose$, where $\LdM$ is the diagonal matrix of $n$ eigenvalues, and $\QM$ contains the corresponding eigenvectors in its columns. 
Then, the embedding matrix $\ZM_\VS$ would be $\QM\LdM^{1/2}$, so that $\ZM_\VS\ZM_\VS\transpose$ approximates $\PsM$ in Eq. \eqref{eq:nodeObj}. Note that, $\ZM_\VS$ has only $k$ columns. To satisfy this, we only take the $k$ leading eigenvalues (forming $\LdM_\Psi$), and let $\QM_k$ contain the corresponding eigenvectors. Then, we get the node embeddings   
\begin{equation}\label{eq:ZV} \textstyle
\ZM_\VS=\QM_k \LdM_\Psi^{1/2}, \textrm{ where }\QM_k\in\RN^{n\times k},\ \LdM_\Psi\in\RN^{k\times k}.
\end{equation}

Similarly, to derive the hyperedge embeddings $\ZM_\ES$  in Eq. \eqref{eq:edgeObj}, the base method  first gets $\PiM^{\prime (T)}$ with at most $T$ steps, and then computes the approximate \heprox matrix $\PsM_T'$.
A note is that we are only interested in deriving embeddings for hyperedges in $\ES$, while the attribute-based hyperedges in $\ES_K$ are constructed just to incorporate the attributes. Thus, we only focus on factorizing the $m\times m$ part of $\PsM'$ (denoted as $\PsM_\ES'$). Although these  attribute-based hyperedges are excluded from $\PsM_\ES'$, the attribute information is actually taken into account in $\PsM_\ES'$ via the multi-hop random walk. Then, after factorizing  $\PsM_\ES'=\QM' \LdM' \QM^{\prime\transpose}$, we take the first $k$ leading eigenvalues (forming $\LdM_\Psi'$) and the corresponding eigenvectors (forming $\QM_k'$) to fit the dimension of embeddings $k$. Finally, we can derive the hyperedge embeddings $\ZM_\ES=\QM_k'\LdM_\Psi^{\prime 1/2}$.

\begin{algorithm}[!t]
\caption{\oursbase} \label{alg:base} 
\small
\KwIn{Hypergraph incidence matrix $\HM_0\in\RN^{m\times n}$ and attribute matrix $\XM\in\RN^{n\times q}$, embedding dimension $k$, $K, \alpha, T$.}
 $\HM, \DM_v, \DM_e, \WM\gets\aug(\HM_0,\XM,K)$\; 
 $\PM\gets\DM_v^{-1}\HM\transpose \WM \DM_e^{-1} \HM,\ \PiM^{(0)}\gets\IM_n$ \tcp*{Eq. \eqref{eq:hgrw-matrix}}
 \For{$t\gets 1, \dots, T$}{ 
  $\PiM^{(t)} \gets \alpha \IM_n + (1-\alpha) \PiM^{(t-1)}\PM$ \tcp*{Eq. \eqref{eq:rwr-recursive}}
  }
 $\PsM_T \gets \tlog\ewise\left(\vol(\HGA)\PiM^{(T)}\DM_v^{-1}\right)$ \; 
 $\LdM_\Psi, \QM_k \gets \eig(\PsM_T, k)$\;
 $\ZM_\VS\gets\QM_k\LdM_\Psi^{1/2}$ \tcp*{Eq. \eqref{eq:ZV}}
 $\PM'\gets\DM_e^{-1}\HM\DM_v^{-1}\HM\transpose\WM,\ \PiM^{\prime(0)}\gets\IM_{m+n}$\;
 \For{$t\gets 1, \dots, T$}{ 
  $\PiM^{\prime(t)} \gets \alpha \IM_{m+n} + (1-\alpha) \PiM^{\prime(t-1)} \PM'$\;
  }
 $\PsM_T' \gets \tlog\ewise\left(\vol(\HGA)\PiM^{\prime(T)}\DM_e^{-1}\WM^{-1}\right)$\;
 $\PsM_\ES'\gets\PsM_T'[1:m+1, 1:m+1]$\;
 $\LdM_\Psi', \QM_k' \gets\eig(\PsM_\ES', k)$\;
 $\ZM_\ES\gets\QM_k'\LdM_\Psi^{\prime1/2}$\;
\Return $\ZM_\VS$, $\ZM_\ES$\; 
\end{algorithm}

The pseudocode of this method is presented in Algorithm \ref{alg:base}. After constructing the extended hypergraph $\HGA$ by Line 1, \oursbase first simulates the random walk  processes on $\HGA$ for $T$ iterations in Lines 2-4, leading to  $\PiM^{(T)}$. The approximate \nprox matrix $\PsM_T$ is computed in Line 5, and the node embedding matrix $\ZM_\VS$ is derived by factorizing $\PsM_T$ in Lines 6-7. For the eigendecomposition in Line 6, we adopt an implementation based on Lanczos iterations, which solves the leading eigenpairs via a limited number of matrix-vector multiplications. Then, Lines 8-14 basically repeat the embedding procedures on the dual hypergraph $\HGA'$ to acquire the hyperedge embeddings $\ZM_\ES$, except that Line 12 removes the last $n$ columns and rows to exclude the attribute-based hyperedges from $\HM$. 

In the experiments, \oursbase demonstrates strong effectiveness, but falls short in scalability. To analyze, Lines 3-4 and Lines 9-10 dominate its time complexity, with time $O(n^2\avgd^2)$ and $O((m+n)^2\avgd^2)$, respectively. Lines 5-6 and Lines 11-13 also incur quadratic time while handling $\PsM_T$ and $\PsM_\ES$. 
Thus, \oursbase has an overall time complexity of $O((m+n)^2\avgd^2)$, and the space complexity is $O((m+n)^2)$, due to the materialization of the $(m+n)\times(m+n)$ matrix $\PiM'$ and the $n\times n$ matrix $\PiM$.
The high complexity stems from the element-wise $\tlog^\circ(\cdot)$ function, which prevents separate factorization of $\PiM$ and $\DM^{-1}$, forcing materialization of $\PiM^{(t)}$ for $t\in [1, T]$. Moreover, $\PiM^{(t)}$ and $\PiM^{\prime (t)}$ grow dense after iterations, exacerbating scalability issues. To overcome these limitations, we design \ours in Section \ref{sec:finalmethod}.

\begin{figure}[!t]
\centering
\includegraphics[width=1\columnwidth]{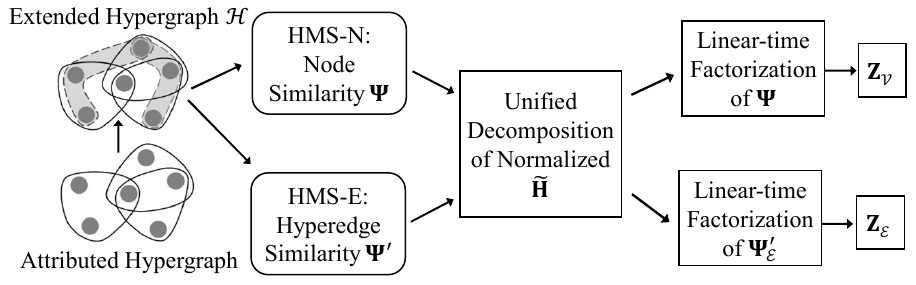}
\vspace{-6mm}
\caption{Overview of the \ours algorithm.} \label{fig:overview}
\vspace{-4mm}
\end{figure}

\section{The \ours Method}\label{sec:finalmethod}
As explained, materializing the dense \nprox and \heprox matrices $\PiM^{(T)}$ and $\PiM^{\prime(T)}$ is computationally expensive. The non-linearity in the definitions of \nprox and \heprox complicates straightforward matrix factorization of the transition and incidence matrices used to construct the similarity matrices.

To solve these difficulties, we develop \ours, an efficient method to produce high-quality \AHE results, with a complete pipeline outlined by Figure \ref{fig:overview}.
\revision{The key ideas are two-fold. First, we analyze the shared core computations of \nprox and \heprox matrices, enabling a unified matrix decomposition procedure for node and hyperedge embedding objectives (Section \ref{sec:unifycomp}). Second, we introduce approximation techniques to efficiently generate node and hyperedge embeddings in linear time, avoiding the materialization of dense \nprox and \heprox similarity matrices, with theoretical guarantees for the approximations (Section \ref{sec:approxNonliear}).}
In Section \ref{sec:finalAlg}, we present the algorithmic details.

\eat{\jieming{we unify the computation of node and hyperedge embeddings.} \yiran{added}
\jieming{challenges, goals, ideas how to solve these.} \yiran{added}
In this section, we propose the \ours algorithm (i.e., Scalable Hypergraph Attribute-augmented Similarity-based Embedding) to overcome the scalability limitations of \oursbase, while unifying the computation of node and hyperedge embeddings. The key challenges are the quadratic time and space demands of materializing dense RWR matrices and factorizing similarity matrices. Our goal is to devise an efficient framework that circumvents these bottlenecks with minimal loss of similarity preservation. To this end, we reformulate the RWR computation to avoid materializing $\PiM^{(T)}$ and $\PiM^{\prime (T)}$, apply truncated SVD to efficiently derive node and hyperedge embeddings in a unified process, and employ Polynomial Tensor Sketch (PTS)~\cite{hanPolynomialTensorSketch2020} to approximate $\tlog$ efficiently.}

\subsection{Unify \nprox and \heprox Computations}\label{sec:unifycomp}

\revision{
    The key strategy is to identify the shared core computations of the \nprox and \heprox matrices, and perform early matrix decomposition to avoid materializing \(\PiM^{(T)}\) and \(\PiM^{\prime(T)}\), thereby improving efficiency.
}
For node embedding, to derive the \nprox matrix $\PsM_T$, according to Eq. \eqref{eq:rwr-formula} and Eq. \eqref{eq:similarity-matrix}, we need to obtain  $\PiM^{(T)}\DM_v^{-1}$ first, 
\begin{equation*} \textstyle
    \textstyle\PiM^{(T)}\DM_v^{-1}= \sum_{i=0}^{T-1} \alpha(1-\alpha)^{i} \PM^i\DM_v^{-1} + (1-\alpha)^{T}\PM^T\DM_v^{-1}.
\end{equation*}
The main computation here is to get the term  $\PM^i\DM_v^{-1}$ ($i\in[T]$). We reformulate  $\PM^i\DM_v^{-1}$ by plugging in the definition of $\PM$ in Eq. \eqref{eq:hgrw-matrix} as follows, and obviously $\PM^i\DM_v^{-1}$ is symmetric. 
\begin{equation}\label{eq:nodeSimAnalysis} \textstyle
    \textstyle\PM^{i}\DM_v^{-1}=\DM_v^{-1/2}\left(\HL\transpose \HL\right)^i\DM_v^{-1/2},
\end{equation}
where $\HL = \WM^{1/2}\DM_e^{-1/2}\HM\DM_v^{-1/2}$ and $\HL\in\RN^{(m+n)\times n}$.

For hyperedge embedding, similarly, to derive the \heprox matrix $\PsM_T'$, according to the formulation in Section \ref{sec:edgeSim}, we need to obtain $\PiM^{'(T)}\DM_e^{-1}\WM^{-1}$, which relies on a  recurring symmetric matrix $\PM^{\prime i}\DM_e^{-1}\WM^{-1}$ that can be decomposed as 
\begin{equation}\label{eq:edgeSimAnalysis} \textstyle
    \PM^{\prime i}\DM_e^{-1}\WM^{-1}=\DM_e^{-1/2}\WM^{-1/2}\left(\HL\HL\transpose\right)^i\DM_e^{-1/2}\WM^{-1/2}.
\end{equation}

Importantly, observe that in Eq. \eqref{eq:nodeSimAnalysis} and Eq. \eqref{eq:edgeSimAnalysis}, they both rely on matrix $\HL\in\RN^{(m+n)\times n}$, which is essentially a normalized version of the incidence matrix $\HM$.
Specifically, both \nprox and \heprox matrices rely on $\HL$ to get either $\left(\HL\transpose \HL\right)^i$ or $\left(\HL\HL\transpose\right)^i$ for $i$ up to $T$.

Note that $\HL$ is sparse since $\HM$ is typically sparse. Therefore, it is fast to decompose $\HL=\UM\SgM\VM\transpose$ by reduced singular value decomposition (RSVD), where $\SgM$ is an $n\times n$ diagonal matrix containing the first $n$ singular values of $\HL$, while $\UM\in\RN^{(m+n)\times n}$ and $\VM\in\RN^{n\times n}$ contain the associated left and right singular vectors as their columns, respectively. 
Then, $\PM^{i}\DM_v^{-1}$ in Eq. \eqref{eq:nodeSimAnalysis} is formulated as  
\begin{equation*} \textstyle
\begin{aligned}
    \PM^{i}\DM_v^{-1}
    &=\DM_v^{-1/2}\left(\VM\SgM\UM\transpose\UM\SgM\VM\transpose\right)^i\DM_v^{-1/2}
    =\DM_v^{-1/2}\left(\VM\SgM^2\VM\transpose\right)^i\DM_v^{-1/2}\\
    &=\DM_v^{-1/2}\VM\SgM^{2i}\VM\transpose\DM_v^{-1/2},
\end{aligned}
\end{equation*}
where the second and third equalities hold since the singular vectors are orthonormal (i.e., $\UM\transpose\UM=\IM_{n}$ and $\VM\transpose\VM=\IM_n$). 

Accordingly, the \nprox matrix $\PsM_T$ can be written as
\begin{equation*} \textstyle
\small
\begin{aligned}\textstyle
    \PsM_T&=\tlog^\circ\left( \vol(\HGA) \PiM^{(T)}\DM_v^{-1}  \right)\\
    &\textstyle=\tlog^\circ\left[\vol(\HGA) \left(  \sum\limits_{i=0}^{T-1} \alpha(1-\alpha)^{i} \PM^i\DM_v^{-1} + (1-\alpha)^{T}\PM^T\DM_v^{-1}  \right)\right]\\
    &\textstyle=\tlog^\circ\!\left[\vol(\HGA)\DM_v^{-\frac{1}{2}}\VM \left(  \sum\limits_{i=0}^{T-1}\! \alpha(1\!-\!\alpha)^{i}  \SgM^{2i}  \!+\! (1\!-\!\alpha)^{T} \SgM^{2T}  \right)\VM\transpose\DM_v^{-\frac{1}{2}}\right].
\end{aligned}
\end{equation*}
Denote $\widehat{\SgM} = \sum_{i=0}^{T-1} \alpha(1-\alpha)^{i} \SgM^{2i} + (1-\alpha)^{T}\SgM^{2T}$, and note that $\widehat{\SgM}$ is an $n\times n$ diagonal matrix.
Simplify the above equation, and we get
\begin{equation*} \textstyle
    \PsM_T=\tlog^\circ\!\left[\vol(\HGA)\DM_v^{-\frac{1}{2}}\VM \widehat{\SgM}\VM\transpose\DM_v^{-\frac{1}{2}}\right].
\end{equation*}
This can be reformulated as 
\begin{equation} \textstyle
    \PsM_T=\tlog^\circ\left( \FM\FM\transpose \right)\text{, where }\FM=\sqrt{\vol(\HGA)}\DM_v^{-1/2}\VM\widehat{\SgM}^{1/2}.
\end{equation}
Similarly, omitting the zero-step residual, we approximate the \heprox matrix $\PsM_T'$ as follows \submission{(derivation in~\cite{report})}\report{(derivation in Appendix~\ref{app:proofs})}:
\begin{equation}\label{eq:edge-gram-approx} \textstyle
    \PsM_T' \approx \tlog^\circ\left( \FM'\FM^{\prime\texttt{T}} \right)\text{, where } \FM'=\sqrt{\vol(\HGA)}\DM_e^{-1/2}\WM^{-1/2}\UM\widehat{\SgM}^{1/2}.
\end{equation}

In this way, both similarity matrices $\PsM_T$ and $\PsM_T'$ can be easily constructed from the RSVD results of $\HL$ via $\FM$ and $\FM'$, without the need to compute $\PiM$ and $\PiM'$, both of which require expensive and repeated multiplications of transition matrices $\PM$ and $\PM'$. Thus, we avoid directly materializing $\PsM_T$ or $\PsM_T'$.

\subsection{\nprox and \heprox Approximations} \label{sec:approxNonliear} 
Despite the reformulation, two efficiency issues still remain. 
First,  the reduced SVD of $\HL$ has a prohibitive $O\left(n(m+n)\right)$ time complexity. To improve scalability, we opt for the truncated SVD $\HL\approx \UM_r\SgM_r\VM_r\transpose$, which only keeps the $r$ largest singular values and the corresponding singular vectors. In this SVD, $\SgM_r=\diag(\sigma_1, \dots, \sigma_r)$ is a diagonal matrix where $\sigma_i$ is the $i$-th largest singular value, while $\UM_r\in\RN^{(m+n)\times r}$ and $\VM_r\in\RN^{n\times r}$ are the left and right singular vectors. By replacing $\SgM$, $\UM$ and $\VM$ with the truncated SVD results, we can derive $\widehat{\SgM}_r$, $\FM_r$ and $\FM_r'$, providing approximations based on rank-$r$ Gram matrices for the similarity matrices, where the error is bounded in Theorem \ref{lm:low-rank}, with proof in \submission{the technical report~\cite{report}.}\report{Appendix~\ref{app:proofs}.}

\begin{theorem}\label{lm:low-rank}
    With rank-$r$ matrices $\FM_r=\sqrt{\vol(\HGA)}\DM_v^{-1/2}\VM_r\widehat{\SgM}_r^{1/2}$ and $\FM_r'=\sqrt{\vol(\HGA)}\DM_e^{-1/2}\WM^{-1/2}\UM_r\widehat{\SgM}_r^{1/2}$, we have the following approximation guarantees for $\PsM_T$ and $\PsM_T'$ ($\,T\ge 1$), where $C$ and $C'$ are constants for a given $\HGA$:
\begin{equation*} \textstyle
\small
\begin{aligned} \textstyle
    \left\Vert\tlog\ewise\left(\FM_r\FM_r\transpose\right)-\PsM_T\right\Vert_F^2
    &\leq C^2\sum_{i=r+1}^{n}\widehat{\SgM}[i,i]^2,\\
    \left\Vert\tlog\ewise\left(\FM_r'\FM_r^{\prime\texttt{T}}\right)-\PsM_T'\right\Vert_F^2
    &\leq (C')^2\left[\sum_{i=r+1}^{n}\widehat{\SgM}[i,i]^2+\alpha^2m\right].
\end{aligned}
\end{equation*}
\end{theorem}


The second challenge arises from the quadratic time and space costs of computing the matrix multiplication $\FM_r\FM_r\transpose$ and applying the subsequent element-wise $\tlog^\circ(\cdot)$ function. To solve this issue, we employ the polynomial tensor sketch (PTS) technique~\cite{hanPolynomialTensorSketch2020} to approximate $\tlog^\circ(\FM_r\FM_r\transpose)$ with $\Gamma=\YM\ThM\YM\transpose$, where $\YM\in\RN^{n\times(\tau b+1)}$ contains the tensor sketches and the diagonal $\ThM$ encodes polynomial coefficients. With PTS, we can efficiently bypass direct matrix materialization. Specifically, we first generate tensor sketches using polynomial degree $\tau$ and sketch dimension $b$, leveraging count-sketch matrices and recursive calculations via the fast Fourier transform. Then, we estimate polynomial coefficients through regression with sample size $c$ to obtain the full approximation.
To analyze, the PTS method takes only linear time $O(n)$ in total, including $O(\tau nr)$ for count-sketch generation, $O(\tau n b)$ for fast Fourier transform and its inverse, and $O(ncr)$ for fitting $\tlog\ewise(\cdot)$ via regression. Moreover, the approximation error of PTS is bounded by Lemma \ref{lm:pts} for our approximation based on the theory in~\cite{hanPolynomialTensorSketch2020}. 
\begin{lemma} \label{lm:pts} If $\left|\tlog(x)-\sum_{i=0}^\tau x^i\right|\leq \epsilon$ for some $\epsilon>0$ in a closed interval containing all entries of $\FM_r\FM_r\transpose$, the PTS $\GmM=\YM\ThM\YM\transpose$ satisfies $  \textstyle
        \mathbb{E} \left\Vert f\ewise(\FM_r\FM_r\transpose)\!-\!\GmM\right\Vert_F^2\!\leq\! 2n^2\epsilon^2+\!\sum_{i=1}^\tau\! \frac{2\tau(2\!+\!3^i) (\ThM[i,i])^2}{b}\!\left[\!\sum_{j=1}^n\!\left\|\FM_r[j,\!:]\right\|_F^{2i}\right]^2\!\!.$
\end{lemma}


Although $\GmM=\YM\ThM\YM\transpose$ resembles the eigendecomposition  of $\GmM$, we cannot directly use $\YM$ and $\ThM$ to derive embeddings, since the dimension of $\YM$, $n\times(\tau b+1)$, does not agree with the $k$-dimensional embedding space. To obtain a $k$-dimensional decomposition efficiently, avoiding the quadratic complexity of standard factorization, we apply the Lanczos method for eigendecomposition. This method computes the $k$ leading eigenpairs of $\GmM$ by iteratively applying the linear operator $\LS(\ve)=\GmM \ve=\YM\left(\ThM\left(\YM\transpose\ve\right)\right)$, which multiplies a vector $\ve$ by $\GmM$ with complexity linear to $n$.   
Consequently, the matrix $\LdM_\GmM$ contains the $k$ largest-magnitude eigenvalues of $\GmM$, and $\QM_{\GmM}$ comprises their corresponding eigenvectors as columns, yielding a factorization: $\QM_{\GmM}\LdM_\GmM\QM_{\GmM}\transpose$. Finally, we get the node embeddings as 
\begin{equation}\label{eq:ZV_shape}
    \textstyle \ZM_\VS=\QM_{\GmM}\LdM_{\GmM}^{1/2}.
\end{equation}


Following a similar process with details omitted, we can approximate the hyperedge similarity matrix $\tlog^\circ(\FM_r^{\prime}{\FM_r'}\transpose)$ with $\GmM'$, decomposed as $\QM_{\GmM}'\LdM_\GmM'{\QM_\GmM'}\transpose$, and derive the hyperedge embeddings 
\begin{equation} \label{eq:ZE_shape} \textstyle 
    \ZM_\ES=\QM_{\GmM}'\LdM_{\GmM}^{\prime 1/2}.
\end{equation}

\subsection{\ours Algorithm Details}\label{sec:finalAlg}
\begin{algorithm}[!t]
\caption{\ours} \label{alg:ahe} 
\small
\KwIn{Hyperedge incidence matrix $\HM_0\in\RN^{m\times n}$, node attribute matrix $\XM\in\RN^{n\times q}$, embedding dimension $k$, algorithm parameters $K,r,T,\alpha,\tau,b,c$.}
 $\HM, \DM_v, \DM_e, \WM\gets\aug(\HM_0,\XM,K)$\; 
 $\HL \gets \WM^{1/2}\DM_e^{-1/2}\HM\DM_v^{-1/2}$ \tcp*{Eq. \eqref{eq:nodeSimAnalysis}}
 $\UM_r, \SgM_r, \VM_r\gets \svds\left(\HL, r\right)$\;
 $\widehat{\SgM}_r\gets \IM_r$\;
  \For{$i\gets 1, \dots, T$}{ 
  $\widehat{\SgM}_r \gets \alpha \IM_r + (1-\alpha) \SgM_r^2 \widehat{\SgM}_r$\;
  }
  $\FM_r\gets\sqrt{\vol(\HGA)}\DM_v^{-1/2}\VM_r\widehat{\SgM}_r^{1/2}$ \tcp*{Theorem \ref{lm:low-rank}}
  $\YM,\ThM\gets\pts(\FM_r, \tlog, \tau, b, c)$\;
Linear operator $\LS(\ve)=\YM\left(\ThM\left(\YM\transpose\ve\right)\right)$\;
 $\LdM_\GmM, \QM_\GmM\gets\eigsh (\LS, k)$\tcp*{$\texttt{eigen}\left(\YM\ThM\YM\transpose, k\right)$}
 $\ZM_\VS\gets\QM_\GmM \LdM_\GmM^{1/2}$ \tcp*{Eq. \eqref{eq:ZV_shape}}
  $\FM_r'\gets\sqrt{\vol(\HGA)}\DM_e^{-1/2}\WM^{-1/2}\UM_r\widehat{\SgM}_r^{1/2}$\;
 $\YM',\ThM'\gets\pts\left(\FM_r'[1:m+1,:], \tlog, \tau, b, c\right)$\;
Linear operator $\LS'(\ve)=\YM'\left(\ThM'\left(\YM^{\prime\texttt{T}}\ve\right)\right)$\;
 $\LdM_\GmM',\QM_\GmM'\gets\eigsh \left(\LS', k\right)$\tcp*{$\texttt{eigen}\left(\YM'\ThM'\YM^{\prime\texttt{T}}, k\right)$}
  $\ZM_\ES\gets\QM_\GmM' \LdM_\GmM^{
\prime 1/2}$ \tcp*{Eq. \eqref{eq:ZE_shape}}
\Return $\ZM_\VS, \ZM_\ES$\;
\end{algorithm} 

With the aforementioned techniques, we can derive node and hyperedge embeddings efficiently without materializing dense similarity matrices. The pseudocode of \ours, our proposed method for attributed hypergraph embedding, is presented in Algorithm \ref{alg:ahe}.

\stitle{Algorithm}
After constructing the attribute-extended hypergraph $\HGA$ at Line 1 by Algorithm \ref{alg:augment} (Line 1), we get the incidence matrix $\HM$, the degree matrices $\DM_v, \DM_e$ and the weight matrix $\WM$. Then, we obtain the normalized hypergraph incidence matrix $\HL$ and decompose it into $\UM_r, \SgM_r$, and $\VM_r$ via the rank-$r$ {TruncatedSVD} (Lines 2-3). We calculate $\widehat{\SgM}_r$ from the $r$ largest singular values of $\HL$ in Lines 4-6. Then we derive the embeddings for nodes and hyperedges. 

For node embeddings (Lines 7-11), we first derive $\FM_r$ by its definition in Theorem \ref{lm:low-rank}, and the node similarity matrix becomes $\tlog^\circ\left( \FM_r\FM_r^{\transpose} \right)$. To compute this $\tlog^\circ\left( \FM_r\FM_r^{\transpose} \right)$ function efficiently, we derive its polynomial tensor sketches $\YM$ and $\ThM$. Finally, the node embeddings $\ZM_{\VS}=\QM_\GmM\LdM_\GmM^{1/2}$ are obtained by factorizing $\GmM=\YM\ThM\YM\transpose$ into $\QM_\GmM\LdM_\GmM\QM_\GmM\transpose$ via the Lanczos technique~\cite{lehoucqDeflationTechniquesImplicitly1996}. 

Following a similar process, we can derive the hyperedge embeddings $\ZM_{\ES}=\QM_\GmM'\LdM_\GmM^{\prime1/2}$ (Lines 12-16), except that in Line 13 we only generate sketches for the first $m$ rows of $\FM_r$, since we are only interested in the embeddings of the original hyperedges.

\stitle{Complexity} With the above approximation techniques, our proposed \ours algorithm has a much lower complexity than the base method. To analyze, we first consider the basic steps. Specifically, invoking \aug to derive the matrices in Line 1 takes $O(n\log n+nqK)$ time. The multiplication of matrices in Line 2 costs only linear time, since $\WM$, $\DM_v$, and $\DM_e$ are diagonal and $\HM$ is a sparse matrix with $n\avgd+nK$ nonzero entries. The \svds technique in Line~3 involves a bounded number of matrix-vector multiplications on $\HL$, and hence incurs $O(n\avgd+nK)$ time complexity. Then, the calculation of $\widehat{\SgM}_r$ in Lines 4-6 takes a negligible $O(Tr)$ time. As can be seen, the common steps for node and hyperedge embedding only take linear time in total. 
To derive node embeddings, we compute $\FM_r$ in Line 7, which takes $O(nr)$ time. Next, recall from Section \ref{sec:approxNonliear} that the approximation via the \pts method in Line 8 finishes in linear time $O(n)$. Regarding the Lanczos method in Lines 9-10, the linear operator $\LS(\cdot)$ executes in linear time and is applied a constant $O(1)$ number of times. Hence, its computation remains linear in time. Finally, the time to obtain node embeddings is $O(n)$. By similar arguments, we can conclude that the time to obtain hyperedge embeddings is $O(m+n)$.
To summarize, the overall time complexity of \ours is $O(n\log n+n\avgd +nq+m)$, or simply $O(n\log n+m)$ as $q$ and $\avgd$ can be considered constant. Moreover, the memory overhead of \ours is $O(n\avgd+nq+m)$, which is linear in the size of the input $\ahg$, since all involved matrices are either sparse or low-dimensional. 

\revision{
\stitle{Discussion} \ours achieves substantial speedup at the cost of approximation errors, compared to \oursbase, which directly computes and factorizes the similarity matrices. Experiments show that on small datasets, the performance of \ours and \oursbase is similar, though \oursbase is often slightly better. However, \oursbase cannot scale to large datasets, while \ours consistently outperforms existing methods in efficiency and effectiveness. 
Hence, the efficiency gain achieved by \ours is well worth the approximation trade-offs.

}


\section{Experiments}\label{sec:expriment}
After providing the experimental settings in  Section \ref{exp:setup}, we report the performance of node embedding on node classification task in Section \ref{exp:node_embedding_performance} and on hyperedge link prediction task in Section \ref{exp:link-prediction}, and the performance of hyperedge embedding on hyperedge classification task in Section \ref{exp:edge_embedding_performance}.
The efficiency results and experimental analysis are in Section \ref{exp:efficiency} and Section \ref{exp:abaltion}. 

\subsection{Experimental Setup}\label{exp:setup}

\begin{table}[!t]
  \caption{Dataset Statistics.}
  \label{tab:datasets}
  \vspace{-3mm}
  \centering
    \resizebox{0.73\columnwidth}{!}{%
    \setlength{\tabcolsep}{3pt}
    \renewcommand{\arraystretch}{0.9} 
  \begin{tabular}{=c+c+c+c+c+c+c}
    \toprule
    Dataset     & $n$ & $m$ &$\avgd$ & $\avge$ &$q$  &$\ell$  \\ \midrule
    DBLP-CA &2,591 &2,690 &2.39 &2.31 &334 &4\\
    Cora-CA &2,708 &1,072 &1.69 &4.28 &1,433 &7\\
    Cora-CC &2,708 &1,579 &1.77 &3.03 &1,433 &7\\
    Citeseer &3,312 &1,079 &1.04 &3.20 &3,703 &6\\
    \newrow
    Mushroom & 8,124 &298 &5.0 &136.3 &126 &2\\
    20News &16,242 &100 &4.03 &654.5 &100 &4\\
    DBLP &41,302 &22,263 &2.41 &4.45 &1,425 &6\\
    \newrow
    Recipe &101,585 &12,387 &25.2 &206.9 &2,254 &8\\
    Amazon &2,268,083 &4,285,295 &32.2 &17.1 &1,000 &15\\
    MAG-PM &2,353,996 &1,082,711 &7.34 &16.0 &1,000 &22\\    
    \bottomrule
  \end{tabular}
  }
  \vspace{-3mm}
\end{table}

\begin{table*}[!t]
\centering
\caption{Node classification performance. The best three are in gray with darker shades indicating better performance.}\vspace{-3mm}

\resizebox{0.98\textwidth}{!}{
\renewcommand{\arraystretch}{0.95}
\setlength{\tabcolsep}{3pt}
\begin{tabular}{|c|cc|cc|cc|cc|BB|cc|cc|BB|cc|cc|B|}
\hline
\multirow{2}{*}{\bf{Method}} & \multicolumn{2}{c|}{\bf{DBLP-CA}} & \multicolumn{2}{c|}{\bf{Cora-CA}} & \multicolumn{2}{c|}{\bf{Cora-CC}} & \multicolumn{2}{c|}{\bf{Citeseer}}  & \multicolumn{2}{c|}{\revision{\bf{Mushroom}}} & \multicolumn{2}{c|}{\bf{20News}} & \multicolumn{2}{c|}{\bf{DBLP}} & \multicolumn{2}{c|}{\revision{\bf{Recipe}}} & \multicolumn{2}{c|}{\bf{Amazon}} & \multicolumn{2}{c|}{\bf{MAG-PM}} & \multirow{2}{*}{\bf{Rank}} \\ \cline{2-21}

&MiF1 & MaF1 &MiF1 & MaF1 &MiF1 & MaF1 &MiF1 & MaF1 &MiF1 & MaF1 &MiF1 & MaF1 &MiF1 & MaF1 &MiF1 & MaF1 &MiF1 & MaF1 &MiF1 & MaF1 &\\ \hline

\hypertovec  & 0.446 & 0.410 & 0.412 & 0.365 & 0.493 & 0.460 & 0.311 & 0.258 &-  &-        & - & - & 0.702 & 0.672 & - & - & - & - & - & - & 8.9 \\
\pane        & 0.671 & 0.651 & 0.516 & 0.456 & 0.508 & 0.491 & 0.443 & 0.399 &0.910 &0.909        & 0.566 & 0.464 & 0.750 & 0.734 & - & - & - & - & \third{0.378} & \third{0.230} & 6.9 \\
\aneci       & 0.683 & 0.661 & 0.625 & 0.582 & 0.453 & 0.367 & 0.454 & 0.399 &0.914	&0.913       & 0.694 & 0.589 & - & - & - & - & - & - & - & - & 7.3 \\
\conn        & 0.756 & 0.744 & 0.684 & 0.640 & 0.637 & 0.577 & 0.626 & 0.563 &-  &-       & - & - & 0.828 & 0.814 & - & - & - & - & - & - & 6.0 \\
\villain     & 0.462 & 0.439 & 0.457 & 0.412 & 0.484 & 0.490 & 0.301 & 0.272 &0.984 &0.984 & 0.730 & 0.645 & 0.692 & 0.657 & - & - & - & - & - & - & 7.6\\
\anchorgnn   & 0.275 & 0.196 & 0.239 & 0.096 & 0.254 & 0.095 & 0.195 & 0.114 &0.854 &0.853       & 0.545 & 0.429 & 0.271 & 0.071 &\third{0.379}	&\third{0.069}  & \third{0.310} & \third{0.032} & 0.252 & 0.018 & 9.2 \\
\biane       & 0.705 & 0.682 & 0.716 & 0.683 & 0.652 & 0.625 & \third{0.644} & \third{0.579} &0.969 &0.969 & - & - & \third{0.853} & \third{0.843} & - & - & - & - & - & -& \third{5.0} \\
\tricl       & 0.787 & 0.778 & 0.702 & 0.677 & \third{0.668} & \third{0.646} & 0.540 & 0.487    &0.978 &0.978      & 0.761 & 0.722 & - & - & - & - & - & - &- &- & \third{5.0} \\
\hypeboy     & \third{0.812} & \third{0.789} & \third{0.725} & \third{0.688}  & 0.627 & 0.584 & 0.476 & 0.420 &0.970	&0.970  & - & - & - & - & - & - & - & - &- &- & 5.8 \\
\netmf       & 0.536 & 0.514 & 0.518 & 0.458 & 0.527 & 0.513 & 0.324 & 0.281 &\third{0.987}	&\third{0.987}        & \third{0.766} & \third{0.733} & 0.744 & 0.721 & - & - & - & - &- &- & 6.2 \\
\lightne     & 0.545 & 0.519 & 0.520 & 0.469 & 0.533 & 0.514 & 0.342 & 0.295 &0.959	&0.959        & 0.700 & 0.646 & 0.733 & 0.712 &\second{0.382}	&\second{0.099} & \second{0.443} & \second{0.210} & \second{0.603} & \second{0.353} & 6.0 \\

\textit{\oursbase} &\first{0.836}	&\first{0.828}	&\first{0.777}	&\first{0.754}	&\first{0.753}	&\first{0.732}	&\first{0.693}	&\first{0.628}	&\second{0.997}	&\second{0.997}	&\first{0.801}	&\first{0.775}	&\first{0.898}	&\first{0.894} & - & - & - & - &- &- & \second{2.1} \\

\textbf{\ours} & \second{0.824} & \second{0.816} & \second{0.753} & \second{0.732} & \second{0.742} & \second{0.720} & \second{0.690} & \second{0.622}    &\first{0.999}	&\first{0.999}          & \second{0.786} & \second{0.748} & \second{0.867} & \second{0.859} &\first{0.630} &\first{0.236} & \first{0.718} & \first{0.396} & \first{0.698} & \first{0.451} & \first{1.6} \\
\hline
\end{tabular}
}
\label{tab:nc_performance}
\vspace{-2mm}
\end{table*}


\noindent\textbf{Datasets.} Table \ref{tab:datasets} summarizes the statistics of attributed hypergraphs used in our experiments, including the number of nodes ($n$) and hyperedges ($m$), the average node degree ($\avgd$), \revision{the average hyperedge size ($\avge$),} the dimension of node attributes ($q$), and the number of ground-truth class labels ($\ell$). DBLP-CA, Cora-CA, Cora-CC, Citeseer, and DBLP are benchmark datasets in~\cite{yadati2019hypergcn}. \revision{Mushroom and 20News are from \cite{chienYouAreAllSet2021}, and   Recipe   is from \cite{liSHARESystemHierarchical2022}.}  Amazon and MAG-PM are million-scale from~\cite{liEfficientEffectiveAttributed2023}. In DBLP-CA, Cora-CA, and MAG-PM, nodes represent publications, and hyperedges link publications by the same author. In  DBLP, nodes are authors, and hyperedges connect co-authors of a publication. Cora-CC and Citeseer are co-citation datasets where hyperedges group publications cited together. Nodes in these datasets have textual attributes from abstracts, with class labels indicating research areas. 
\revision{The Mushroom dataset forms hyperedges by connecting mushrooms (nodes) with the same traits. A mushroom has a one-hot binary attribute vector from categorical features and is labeled as edible or poisonous.}
The 20News dataset forms hyperedges by shared keywords, using TF-IDF vectors as node attributes and topics as labels. 
\revision{Recipe is a recipe-ingredient hypergraph with bag-of-words attributes from instruction texts and dense hyperedge connections.} 
In Amazon, nodes are products, hyperedges connect products reviewed by the same user, and attributes come from metadata, with categories as labels. Hyperedges lack labels, so we assign each the most frequent node label. For example, an author hyperedge in Cora-CA takes the predominant research area among its publications, while in Amazon, a user hyperedge adopts the most common product category. 

\stitle{Baselines} For \textit{node embedding} evaluation, we compare \ours against 11 baselines in total, including the hypergraph embedding approach \hypertovec\cite{huangHyper2vecBiasedRandom2019}, and three attributed graph embedding approaches (\ie \pane~\citep{yangPANEScalableEffective2023}, \aneci\cite{liu2022robust}, and \conn\cite{tan2023collaborative}), which are applied to reduced graphs derived from the clique expansion of the hypergraph. Also, we consider two bipartite graph embedding techniques \anchorgnn~\cite{wuBillionScaleBipartiteGraph2023} and \biane~\cite{huangBiANEBipartiteAttributed2020a} that are applied to a bipartite graph where hyperedges are treated as a distinct set of nodes separate from the original nodes. Finally, we include three self-supervised learning baselines (\villain~\cite{leeVilLainSelfSupervisedLearning2024}, \tricl\cite{leeImMeWere2023}, and \hypeboy\cite{kimHypeBoyGenerativeSelfSupervised2023}), with \tricl and \hypeboy targeted for attributed hypergraph embedding, and matrix factorization approaches on the general graph, \netmf~\cite{qiuNetworkEmbeddingMatrix2018a} and \lightne~\cite{qiuLightNELightweightGraph2021}. For \textit{hyperedge embedding} evaluation, we also compare these baselines, among which bipartite graph embedding methods (\anchorgnn and \biane) can produce embeddings for two parts as node and hyperedge embeddings, respectively. The remaining methods compute a hyperedge embedding by averaging the node embeddings in the hyperedge. In addition, we also compare \ours with the base method in Section \ref{sec:basemethod} for effectiveness.

\stitle{Implementation} \revision{On all datasets, \ours and \oursbase have the identical parameter settings: $K=10$, $\beta=1.0$, $\alpha = 0.1$, $T=10$. 
For all datasets, \ours performs approximation with $r=32$, $\tau=3$, $b=128$, and $c=10$, except Mushroom with $r=16$.}  We fix the output node and hyperedge embedding dimension $k$ to 32 for all approaches. The parameters for all tested baselines are configured according to their respective papers. Our method  \ours,  along with most baselines, is implemented in Python, except for the C++ competitor \lightne. 

\stitle{Evaluation} We conduct experimental evaluations on a Linux computer with an Intel Xeon Platinum 8338C CPU, an NVIDIA RTX 3090 GPU, and 384 GB of RAM, where a maximum of 16 CPU threads are available. The methods \aneci, \conn, \villain, \anchorgnn, \tricl, and \hypeboy benefit from GPU acceleration, while the other methods, including \oursbase and \ours, are executed on the CPU.
We report average results over 10 repeated runs.  
If an approach fails to complete within 24 hours or runs out of memory, it is considered to rank last, and we {record the result as ` - '} in Tables \ref{tab:nc_performance}-\ref{tab:hec_performance}.


\begin{table*}[!t]
\centering
\caption{Hyperedge link prediction performance. The best three are in gray with darker shades indicating better performance. }\vspace{-3mm}

\resizebox{0.98\textwidth}{!}{
\renewcommand{\arraystretch}{0.96}
\setlength{\tabcolsep}{3.5pt}
\begin{tabular}{|c|cc|cc|cc|cc|BB|cc|cc|BB|cc|cc|B|}
\hline
\multirow{2}{*}{\bf{Method}} & \multicolumn{2}{c|}{\bf{DBLP-CA}} & \multicolumn{2}{c|}{\bf{Cora-CA}} & \multicolumn{2}{c|}{\bf{Cora-CC}} & \multicolumn{2}{c|}{\bf{Citeseer}} & \multicolumn{2}{c|}{\revision{\bf{Mushroom}}} & \multicolumn{2}{c|}{\bf{20News}} & \multicolumn{2}{c|}{\bf{DBLP}}  & \multicolumn{2}{c|}{\revision{\bf{Recipe}}}  & \multicolumn{2}{c|}{\bf{Amazon}} & \multicolumn{2}{c|}{\bf{MAG-PM}} & \multirow{2}{*}{\bf{Rank}} \\ \cline{2-21}

&Acc & AUC &Acc & AUC  &Acc & AUC  &Acc & AUC  &Acc & AUC  &Acc & AUC  &Acc & AUC  &Acc & AUC  &Acc & AUC &Acc & AUC &\\ \hline
\hypertovec   & 0.631 & 0.712 & 0.667 & 0.751 & 0.715 & 0.751 & 0.669 & 0.684 &- &- & -   & -   & 0.704 & 0.741 &- &- & -   & -   & -   & -   & 8.1 \\
\pane         & 0.687 & 0.774 & 0.685 & 0.765 & 0.747 & 0.755 & 0.685 & 0.680 &0.930 &0.974 & 0.513 & 0.638 & 0.723 & \third{0.831} &-  &- & -   & -   & \third{0.622} & \third{0.697} & 6.1 \\
\aneci        & 0.704 & 0.797 & 0.695 & 0.778 & 0.753 & 0.836 & \second{0.793} & \third{0.890} &0.947	&0.976  & 0.615 & 0.617 & -   & -   & -   & -   & -   & -   &- &-& 5.4 \\
\conn         & \first{0.797} & \third{0.880} & 0.655 & 0.710 & 0.737 & 0.835 & 0.751 & 0.856 &-  &-  & -   & -   & 0.727 & 0.814 & -   & -   & -   & -   &- &-& 6.3 \\
\villain      & 0.638 & 0.721 & 0.682 & 0.729 & 0.729 & 0.833 & 0.659 & 0.717 &0.905	&0.971 & 0.500 & 0.396 & 0.698 & 0.676 & -   & -   & -   & -   &- &-& 7.4 \\
\anchorgnn    & 0.530 & 0.553 & 0.512 & 0.525 & 0.628 & 0.688 & 0.565 & 0.603 &0.693	&0.822 & 0.515 & 0.403 & 0.516 & 0.522 &\third{0.506}	&\third{0.553} & \third{0.694} & \third{0.773} & 0.484 & 0.476 & 9.1 \\
\biane        & 0.638 & 0.599 & 0.648 & 0.597 & 0.751 & 0.721 & 0.690 & 0.647 &0.941 &0.981 & -   & -   & 0.681 & 0.631 & -   & -   & -   & -   &- &-& 7.9 \\
\tricl        & 0.719 & 0.808 & 0.682 & 0.738 & 0.727 & 0.837 & 0.720 & 0.824 &0.942	&\third{0.988} & 0.615 & 0.858 & -   & -   & -   & -   & -   & -   &- &-& 5.7 \\
\hypeboy      & 0.718 & 0.836 & \third{0.740} & \first{0.843} & \first{0.835} & \second{0.924} & 0.741 & 0.805 &0.937	&0.982 & -   & -   & -   & -   & -   & -   & -   & -   &- &-& \third{5.4} \\
\netmf        & 0.659 & 0.715 & \third{0.740} & 0.793 & 0.722 & 0.736 & 0.643 & 0.617 &0.943	&\third{0.988} & \third{0.755} & \third{0.873} & \third{0.755} & 0.817 & -   & -   & -   & -   &- &-& 5.9 \\
\lightne      & 0.632 & 0.676 & 0.675 & 0.672 & 0.725 & 0.839 & 0.671 & 0.756 &\third{0.954}	&\third{0.988} & 0.535 & 0.658 & 0.696 & 0.703 &\second{0.642} &\second{0.689} & \second{0.732} & \second{0.820} & \second{0.746} & \second{0.793} & 5.8 \\
\oursbase &\second{0.785}	&\first{0.893}	&\second{0.744}	&\third{0.815}	&\third{0.790}	&\third{0.899}	&\third{0.783}	&\second{0.905}	&\second{0.968}	&\second{0.996}	&\second{0.825}	&\first{0.969}	&\second{0.811}	&\second{0.896} & -   & -   & -   & -   &- &-  & \second{2.8} \\

\textbf{\ours} & \third{0.776} & \second{0.890} & \first{0.766} & \second{0.828} & \second{0.807} & \second{0.902} & \first{0.801} & \first{0.916} &\first{0.989} &\first{0.999}  & \first{0.870} & \second{0.956} & \first{0.824} & \first{0.911} &\first{0.763}	 &\first{0.830} & \first{0.909} & \first{0.965} & \first{0.761} & \first{0.798} & \first{1.4} \\
\hline
\end{tabular}
}
\label{tab:lp_performance}
\vspace{-1mm}
\end{table*}


\subsection{Node Classification} \label{exp:node_embedding_performance}

For attributed hypergraphs, node classification seeks to predict class labels using node embeddings. We split datasets into training and test sets, using a 20\%/80\% ratio for most, except Amazon and MAG-PM, where 2\% is allocated for training due to their size. Ten random splits are generated per dataset, and we report average results. Embeddings, derived without accessing label information, are used to train a simple linear classifier on the training set, with performance evaluated on the test set. Classification effectiveness is assessed via Micro-F1 (MiF1) and Macro-F1 (MaF1) scores, where higher values indicate better performance.

Table \ref{tab:nc_performance} shows the results, with the top three performances for each dataset highlighted in gray, using darker shades for better performance. The Rank column indicates the average ranking of each method across all metrics. \ours achieves the best overall rank of 1.6, significantly outperforming the next best competitors, \biane and \tricl, which rank at 5.0. On large datasets like Amazon and MAG-PM, most competitors fail to return results within time and memory limits. Compared to \oursbase from Section \ref{sec:basemethod}, \ours, developed in Section \ref{sec:finalmethod}, is outperformed slightly on small datasets but excels on large ones where \oursbase is inefficient. This highlights the effectiveness of \ours's approximation techniques in maintaining result quality while improving efficiency.
For instance, on Cora-CC, \oursbase and \ours secure the first and second positions, respectively, outperforming the third-ranked \tricl by up to 8.5\% in both MiF1 and MaF1. {On the DBLP-CA, Cora-CA, Citeseer, Mushroom, 20News, and DBLP datasets, \ours improves over the best competitors by 1.2\%, 2.8\%, 4.6\%, 1.2\%, 2.0\%, and 1.4\% in MiF1, and 2.7\%, 4.4\%, 4.3\%, 1.2\%, 1.5\%, and 1.6\% in MaF1, respectively. On the densely connected Recipe, \ours significantly outperforms the best competitor by 24.8\% in MiF1 and 13.7\% in MaF1.} On the large Amazon and MAG-PM, \ours also surpasses the runner-up with margins up to 27.5\% in MiF1 and 18.6\% in MaF1 on Amazon.
Table \ref{tab:nc_performance} demonstrates \ours's excellent performance in node classification, indicating the high quality of node embeddings and the effectiveness of the \nprox objective from Section \ref{sec:problemformulation} and algorithm designs in Section \ref{sec:approxNonliear}.

\begin{table*}[!t]
\centering
\caption{Hyperedge classification performance. The best three are in gray with darker shades indicating better performance. (20News is excluded for lack of suitable labels.)}\vspace{-3mm}

\resizebox{0.93\textwidth}{!}{
\renewcommand{\arraystretch}{0.96}
\setlength{\tabcolsep}{3.5pt}
\begin{tabular}{|c|cc|cc|cc|cc|BB|cc|BB|cc|cc|B|}
\hline
\multirow{2}{*}{\bf{Method}} & \multicolumn{2}{c|}{\bf{DBLP-CA}} & \multicolumn{2}{c|}{\bf{Cora-CA}} & \multicolumn{2}{c|}{\bf{Cora-CC}} & \multicolumn{2}{c|}{\bf{Citeseer}} & \multicolumn{2}{c|}{\revision{\bf{Mushroom}}}  & \multicolumn{2}{c|}{\bf{DBLP}} & \multicolumn{2}{c|}{\revision{\bf{Recipe}}}  & \multicolumn{2}{c|}{\bf{Amazon}} & \multicolumn{2}{c|}{\bf{MAG-PM}} & \multirow{2}{*}{\bf{Rank}} \\ \cline{2-19}

&MiF1 & MaF1 &MiF1 & MaF1 &MiF1 & MaF1 &MiF1 & MaF1 &MiF1 & MaF1 &MiF1 & MaF1 &MiF1 & MaF1 &MiF1 & MaF1 &MiF1 & MaF1 & \\ \hline

\hypertovec &0.569	&0.518 &0.439	&0.370 &0.794 &0.786 &0.589	&0.511 &- &- &0.599	&0.553 &- &- &- &- &- &-  &7.7\\
\pane &0.704	&0.673 &0.515	&0.426 &0.737	&0.725 &0.567	&0.495 & 0.769	&0.765  &0.751	&0.731 &- &- &- &- &0.303	&\third{0.111} &6.4 \\
\aneci &0.685	&0.664 &0.599	&0.529 &0.542	&0.484 &0.534	&0.398 &  0.821	&0.813 &- &- &- &- &- &- &- &- &7.6 \\
\conn &0.809	&0.786 &0.641	&0.587 &0.781	&0.760 &0.689	&0.612 &- &-  &\third{0.837}	&\third{0.815} &- &- &- &- &- &- &{5.6} \\
\villain &0.567	&0.515 &0.459	&0.382 &0.799	&0.790 &0.619	&0.538 & 0.838	&0.835 &0.550	&0.486 &- &- &- &- &- &- &6.4 \\
\anchorgnn &0.307	&0.251 &0.185	&0.142 &0.194	&0.146 &0.201	&0.170 &0.622	&0.602 &0.267	&0.087 &\second{0.454}	&\second{0.078} &\third{0.372}	&\third{0.036} &\third{0.334}	&0.044 &9.1 \\
\biane &0.479	&0.408 &0.241	&0.179 &0.577	&0.506 &0.462	&0.377 &0.767 &0.762 &0.462 &0.342 &-  &- &-  &- &- &- &8.8\\ 
\tricl &0.804	&0.778 &0.646	&0.590 &\third{0.820}	&\third{0.808} &0.659	&0.579 & 0.838	&0.834 &- &-  &- &- &- &- &- &- &5.2 \\
\hypeboy &\third{0.820}	&\third{0.799} &\third{0.719}	&\third{0.662} &0.794	&0.775 &\third{0.728}	&\third{0.630} & 0.835	&0.830 &-  &- &-  &- &- &- &- &- &5.3 \\
\netmf &0.650	&0.607 &0.452	&0.399 &0.797	&0.792 &0.596	&0.530 & \third{0.879}	&\third{0.877} &0.730	&0.699 &- &- &- &- &- &- &5.8 \\
\lightne &0.654	&0.610 &0.458	&0.399 &0.800	&0.791 &0.617	&0.534 &0.847	&0.844 &0.702	&0.673 &\third{0.199}	&\third{0.051} &\second{0.774}	&\first{0.485} &\second{0.502}	&\second{0.179} &\third{4.9} \\

\oursbase &\first{0.858}	&\first{0.838} &\first{0.775}	&\first{0.740} &\first{0.852}	&\first{0.846} &\first{0.770}	&\first{0.684}  &\first{0.926}	&\first{0.925} &\first{0.908}	&\first{0.898} &- &- &- &- &- &- &\second{2.1} \\
\textbf{\ours} &\second{0.854}	&\second{0.836} &\second{0.764}	&\second{0.711} &\second{0.850}	&\second{0.839} &\second{0.756}	&\second{0.669} & \second{0.908}	&\second{0.907} &\second{0.863}	&\second{0.843} &\first{0.668}	&\first{0.236}  &\first{0.823}	&\second{0.429} &\first{0.755}	&\first{0.470} &\first{1.7}\\
 
\hline
\end{tabular}
}
\label{tab:hec_performance}
\vspace{-1mm}
\end{table*}


\subsection{Hyperedge Link Prediction} \label{exp:link-prediction}
Hyperedge link prediction in attributed hypergraphs seeks to identify whether a group of nodes forms a real hyperedge using node embeddings~\cite{patil2020negative, leeVilLainSelfSupervisedLearning2024}. For each dataset, we divide hyperedges into training and test sets, using an 80\%/20\% split for smaller datasets and 98\%/2\% for larger ones like Amazon and MAG-PM. For each real hyperedge, we create a negative counterpart by randomly selecting nodes to match its size. Node embeddings are derived from the training set's real hyperedges and full attribute data, excluding test hyperedges and labels. A linear binary classifier is trained to differentiate real from negative hyperedges, using max-min aggregation of node embeddings as input. This model is tested on the test set, predicting real and negative hyperedges. We repeat this process over 10 random splits, averaging the results. Performance is assessed by accuracy (Acc) and area under the ROC curve (AUC), with higher scores indicating better performance.

\eat{Hyperedge link prediction in attributed hypergraphs aims to determine whether a group of nodes forms a real hyperedge, leveraging the embeddings of these nodes~\cite{patil2020negative, leeVilLainSelfSupervisedLearning2024}. For each dataset, we split the hyperedges into training and test sets, using an 80\%/20\% ratio for smaller datasets and a 98\%/2\% ratio for Amazon and MAG-PM due to their scale. To create test cases, we generate one negative hyperedge per real hyperedge in both sets by randomly sampling nodes from the full node set, matching the size of the corresponding real hyperedge (e.g., if a real hyperedge has 5 nodes, its negative counterpart also has 5 nodes). Node embeddings are computed using only the real hyperedges from the training set and full attribute information, without access to test hyperedges or labels. A linear binary classifier is then trained on the training set to distinguish real from negative hyperedges, taking as input the max-min aggregation (concatenation of maximum and minimum values across dimensions) of the node embedding vectors within each hyperedge, and outputting a binary label. This predictor is evaluated on the test set, predicting real and negative hyperedges. We perform this process across 10 random training/test splits, reporting average results. Performance is measured by accuracy (Acc) and area under the ROC curve (AUC), where higher scores indicate better performance.}

Table \ref{tab:lp_performance} shows that \ours ranks highest overall with a score of 1.4, significantly outperforming the strongest baseline, \hypeboy, which has a rank of 5.4. While \oursbase performs well on smaller datasets, it struggles with larger ones. In contrast, \ours maintains top performance on large datasets like Amazon and MAG-PM, where other methods falter. 
{For instance, on Recipe, \ours exceeds \lightne by 12.1\% in accuracy and 14.1\% in AUC.} 
On the large Amazon dataset, \ours achieves 90.9\% accuracy and 96.5\% AUC, improving by 17.7\% and 14.5\% over the runner-up, \lightne, which scores 73.2\% accuracy and 82.0\% AUC. These results confirm that \ours generates high-quality node embeddings, validating the effectiveness of our proposed node similarity measure and embedding objective.

\eat{Table \ref{tab:lp_performance} displays the results, where \ours achieves the top overall rank of 1.6, significantly outperforming the next best competitor, \hypeboy, with a rank of 5.5. The \oursbase method also demonstrates high effectiveness on small datasets, surpassing existing methods, but cannot scale to large datasets. On the other hand, \ours uniquely maintains the highest quality on large-scale datasets like Amazon and MAG-PM, where most existing methods struggle to scale.
For example, on 20News, \ours surpasses the best competitor \netmf by up to 20\% in Acc, and 10.5\% in AUC. On DBLP, \ours outperforms the best competitor by a margin of 6.9\% in accuracy and 8.0\% in AUC. On the large Amazon dataset, \ours achieves 90.9\% Acc and 96.5\% AUC, improving by 17.6\% and 14.5\% compared with the runner-up, \lightne, with 73.2\% Acc and 82.0\% AUC. The results in Table \ref{tab:lp_performance} validate that our proposed \ours produces high-quality node embeddings, and further support our proposed node similarity measure and embedding objective.}

\subsection{Hyperedge Classification}\label{exp:edge_embedding_performance} 


\input{figure_revision/new_time}

We evaluate hyperedge embeddings using a classification task that predicts a hyperedge's label from its embedding vector. Hyperedges are split into training and test sets with a 20\%/80\% ratio, except for Amazon and MAG-PM, which use a 2\%/98\% split due to their size. Embeddings are computed from the attributed hypergraph without label information. A linear classifier is trained on the training set, using hyperedge embeddings as input and their labels as targets. Performance is assessed on the test set, averaged over 10 random splits, and measured by MiF1 and MaF1. 
Table \ref{tab:hec_performance} shows that \ours ranks first overall with a score of 1.7, significantly outperforming the closest competitor, \lightne, which ranks at 4.9. On large datasets like Amazon and MAG-PM, most competitors fail due to time or memory constraints. Unlike \oursbase, which struggles with larger datasets, \ours excels on both small and large datasets, thanks to its efficient approximation techniques in Section \ref{sec:finalmethod}. 
Compared to the runner-up baseline, \ours improves MiF1 by 4.5\% and MaF1 by 4.9\% on Cora-CA, and by large margins of 21.4\% in MiF1 and 15.8\% in MaF1 on the Recipe dataset. 
On the large MAG-PM dataset, \ours outperforms \lightne by 25.3\% in MiF1 and 29.1\% in MaF1.  This suggests that simply averaging node embeddings, as in baseline methods, is insufficient for hyperedge embedding. The performance of \ours shows the effectiveness of the \heprox similarity objective and approximation techniques in Sections \ref{sec:edgeSim} and \ref{sec:finalmethod}.

\eat{As shown in Table \ref{tab:hec_performance}, \ours achieves the top overall rank of 1.4, significantly surpassing the closest competitor, \conn, which has a rank of 4.9. On larger datasets such as Amazon and MAG-PM, most competitors are unable to produce results within the time or memory constraints. In comparison to the \oursbase method, \ours, developed in Section \ref{sec:finalmethod}, performs similarly on smaller datasets and excels on two larger datasets where \oursbase proves inefficient. This highlights the effectiveness of the approximation techniques in \ours, which maintain result quality while greatly enhancing efficiency.
Against the strongest \conn, \ours improves MiF1 by 3.4\% and MaF1 by 3.7\% on DBLP-CA, and by 4.5\% and 4.9\% on Cora-CA, respectively. Notably, \ours surpasses the top competitor \lightne by a substantial 25.3\% in MiF1 and 29.1\% in MaF1 on the large MAG-PM dataset. As \lightne is evaluated by averaging node embeddings for hyperedge classification, we infer that such straightforward aggregation of node embeddings is inadequate for solving the hyperedge embedding problem. 
In contrast, the consistently outstanding performance of \ours highlights the effectiveness of \heprox similarity objective and approximation techniques in Section \ref{sec:edgeSim} and Section \ref{sec:approxNonliear}.}

\subsection{Embedding Efficiency}\label{exp:efficiency}
Figure \ref{fig:dualemb_time} reports the time of all methods to generate node and hyperedge embeddings across the ten datasets, with each chart’s y-axis showing running time in seconds on a logarithmic scale and stars marking the top-performing competitors in all three tasks. 

Observe that (i) \ours consistently outperforms all competitors in terms of efficiency across all datasets, regardless of the competitors' quality; more importantly, (ii) \ours demonstrates a significant speed advantage over the most effective competitors marked by stars in Figure \ref{fig:dualemb_time}, often being faster by orders of magnitude.
Taking the Citeseer dataset as example, Figure \ref{fig:time-emb-citeseer-cc} shows that \ours is $422.5\times$ faster than \biane (3rd place in node classification), $15.3\times$ faster than \aneci (3rd place in hyperedge link prediction), and $70.2\times$ faster than \hypeboy (3rd place in hyperedge classification), while \ours outperforms all baselines in these tasks.
In Figure \ref{fig:time-emb-20news}, \ours takes just 0.951 seconds compared to 13,536 seconds for \netmf, the runner-up in embedding quality. This is because \netmf operates on a dense clique-expansion graph reduced from the hyperedges, which incurs a quadratic complexity for the factorization-based algorithm. On the million-scale Amazon dataset, \ours is much faster than the competitors, such as \lightne, with an average rank of 6.0, compared to the 1.4 average rank of \ours in Table \ref{tab:nc_performance}.  
These results underscore the combination of high-quality embeddings and excellent efficiency achieved by \ours.

\subsection{Experimental Analysis} \label{exp:abaltion} 

\revision{
\noindent\textbf{Scalability test.} 
We assess scalability on synthetic attributed hypergraphs with the number of nodes $n$ ranging from 2 to 10 million. Each hypergraph is generated as a 3-uniform hypergraph with $n$ hyperedges of size 3~\cite{gaoHGNNGeneralHypergraph2023}, and each node is assigned 100 random binary attributes. Figure~\ref{fig:scalability} shows the time and memory usage of \ours against the scalable baseline \lightne on CPU. \ours exhibits near-linear scalability and outperforms \lightne in both metrics, confirming the complexity analysis in Section~\ref{sec:finalAlg} and demonstrating the efficiency of \ours for large datasets in practice. \definecolor{darkblue}{RGB}{0,0,139}
\definecolor{darkgray176}{RGB}{176,176,176}
\definecolor{darkgreen}{RGB}{0,100,0}
\definecolor{darkorange25512714}{RGB}{255,127,14}
\definecolor{darkred}{RGB}{139,0,0}
\definecolor{forestgreen4416044}{RGB}{44,160,44}
\definecolor{goldenrod}{RGB}{218,165,32}
\definecolor{lightgray204}{RGB}{204,204,204}
\definecolor{lightseagreen0171189}{RGB}{0,171,189}
\definecolor{lightsteelblue161199224}{RGB}{161,199,224}
\definecolor{magenta}{RGB}{255,0,255}
\definecolor{steelblue31119180}{RGB}{31,119,180}
\definecolor{teal2110129}{RGB}{2,110,129}
\definecolor{myred}{HTML}{fd7f6f}
\definecolor{mywhite}{HTML}{D8D8D8}
\definecolor{myorange}{HTML}{D7191C}
\definecolor{myblue}{HTML}{7eb0d5}
\definecolor{mygreen}{HTML}{b2e061}
\definecolor{mypurple}{HTML}{bd7ebe}
\definecolor{myorange}{HTML}{ffb55a}
\definecolor{myyellow}{HTML}{ffee65}
\definecolor{mypurple2}{HTML}{beb9db}
\definecolor{mypink}{HTML}{fdcce5}
\definecolor{mycyan}{HTML}{8bd3c7}
\definecolor{mycyan2}{HTML}{00ffff}
\definecolor{myblue2}{HTML}{115f9a}
\definecolor{myred2}{HTML}{c23728}
\begin{figure}[!t]
\centering
\begin{small}

\begin{tikzpicture}
    \begin{customlegend}[
        legend entries={\ours, \lightne},
        legend columns=4,
        legend style={at={(0.05,1.05)},anchor=north,draw=none,font=\footnotesize, column sep=0.05cm,row sep=0.05cm}]
    \addlegendimage{line width=0.3mm,mark size=3pt,mark=square,color=teal}
    \addlegendimage{line width=0.3mm,mark size=3pt,mark=x,color=orange}
    \end{customlegend}
\end{tikzpicture}
\\[-\lineskip]
\vspace{-4mm}
\subfloat{
\resizebox{0.2\textwidth}{!}{
\begin{tikzpicture}[every mark/.append style={mark size=2pt}]

\definecolor{darkgrey176}{RGB}{176,176,176}
\definecolor{steelblue31119180}{RGB}{31,119,180}

\begin{axis}[
height=\columnwidth/2.8,
width=\columnwidth/2,
axis lines=left,
tick align=outside,
tick pos=left,
x grid style={darkgrey176},
xlabel={$n/10^6$},
xmin=0, xmax=11990000,
xtick style={color=black},
xtick={-2000000,0,2000000,4000000,6000000,8000000,10000000,12000000},
xticklabels={\ensuremath{-}2,0,2,4,6,8,10,12},
scaled x ticks=false,
y grid style={darkgrey176},
ylabel={time (s)},
ymin=-106.00355, ymax=4950.59355,
ytick style={color=black},
every axis y label/.style={at={(current axis.north west)},right=-3mm,above=1mm},
every axis x label/.style={at={{(0.68,0.0)}},right=13mm,below=0mm},
]
\addplot [line width=0.3mm,mark=square,color=teal]
table[row sep=crcr] {
2000000 386.4397\\
4000000 789.1107\\
6000000 1201.6772\\
8000000 1675.7615\\
10000000 2133.2827\\
};
\addplot [line width=0.3mm,mark=x,color=orange]
table[row sep=crcr] {
2000000 708.8777\\
4000000 1527.4336\\
6000000 2432.9860\\
8000000 3233.3741\\
10000000 4470.9281\\
};
\end{axis}

\end{tikzpicture}

}}
\subfloat{
\resizebox{0.2\textwidth}{!}{
\begin{tikzpicture}

\definecolor{darkgrey176}{RGB}{176,176,176}
\definecolor{steelblue31119180}{RGB}{31,119,180}

\begin{axis}[
axis lines=left,
height=\columnwidth/2.8,
width=\columnwidth/2,
tick align=outside,
tick pos=left,
x grid style={darkgrey176},
xlabel={$n/10^6$},
xmin=0, xmax=11990000,
xtick style={color=black},
xtick={-2000000,0,2000000,4000000,6000000,8000000,10000000,12000000},
xticklabels={\ensuremath{-}2,0,2,4,6,8,10,12},
scaled x ticks=false,
y grid style={darkgrey176},
ylabel={RAM (GB)},
ymin=-10.12249, ymax=337.83449,
ytick style={color=black},
every axis y label/.style={at={(current axis.north west)},right=0mm,above=1mm},
every axis x label/.style={at={{(0.68,0.0)}},right=13mm,below=0mm},
]
\addplot [line width=0.3mm,mark=square,color=teal]
table[row sep=crcr] {
2000000 24.6745\\
4000000 48.9083\\
6000000 73.1697\\
8000000 97.5176\\
10000000 121.8819\\
};
\addplot [line width=0.3mm,mark=x,color=orange]
table[row sep=crcr] {
2000000 49.6800\\
4000000 98.7200\\
6000000 149.1000\\
8000000 200.2000\\
10000000 261.7000\\
};
\end{axis}

\end{tikzpicture}

}}

\vspace{-3mm} 
\caption{\revision{Scalability Test.}} \label{fig:scalability}
\vspace{-2mm} 
\end{small} 
\end{figure} 
}

\begin{table}[t]
\centering
\captionsetup[subfloat]{font=small}
\caption{\revision{Approximation Error (MAE).}}
\label{tab:mae}
\vspace{-3.5mm}
\setlength{\tabcolsep}{7pt}     \renewcommand{\arraystretch}{0.9} 
\resizebox{0.65\columnwidth}{!}{\begin{tabular}{|l|c|c|c|c|}
\hline
 & \multicolumn{2}{c|}{\textbf{\nprox}} & \multicolumn{2}{c|}{\textbf{\heprox}} \\
\cline{2-5}
\textbf{Dataset} & {\oursbase} & {\ours} & {\oursbase} & {\ours} \\
\hline
DBLP-CA & {0.0887} & {0.1281} & {0.0795} & {0.2141} \\
Cora-CA & {0.0965} & {0.1384} & {0.0770} & {0.2761} \\
Cora-CC & {0.0970} & {0.1546} & {0.0551} & {0.1714} \\
Citeseer & {0.0927} & {0.1446} & {0.0629} & {0.2096} \\
\hline
\end{tabular}}
\vspace{-2mm}
\end{table}
\revision{
\stitle{Approximation error} 
\ours improves efficiency by introducing acceptable approximation errors compared to \oursbase, which directly computes and factorizes similarity matrices. Table~\ref{tab:mae} quantifies this loss by reporting the mean absolute error (MAE) between normalized \nprox and \heprox matrices and their embedding dot product matrices. Specifically, similarity matrices are normalized by their diagonal mean to align self-similarity scales, and MAE is computed as the difference between the embedding dot product and the similarity matrices. Results show low errors for both methods, with \oursbase achieving slightly lower MAE. This confirms \ours effectively approximates similarity measures with small errors, enabling comparable effectiveness while ensuring efficiency.}

\definecolor{darkblue}{RGB}{0,0,139}
\definecolor{darkgray176}{RGB}{176,176,176}
\definecolor{darkgreen}{RGB}{0,100,0}
\definecolor{darkorange25512714}{RGB}{255,127,14}
\definecolor{darkred}{RGB}{139,0,0}
\definecolor{forestgreen4416044}{RGB}{44,160,44}
\definecolor{goldenrod}{RGB}{218,165,32}
\definecolor{lightgray204}{RGB}{204,204,204}
\definecolor{lightseagreen0171189}{RGB}{0,171,189}
\definecolor{lightsteelblue161199224}{RGB}{161,199,224}
\definecolor{magenta}{RGB}{255,0,255}
\definecolor{steelblue31119180}{RGB}{31,119,180}
\definecolor{teal2110129}{RGB}{2,110,129}
\definecolor{myred}{HTML}{fd7f6f}
\definecolor{mywhite}{HTML}{D8D8D8}
\definecolor{myorange}{HTML}{D7191C}
\definecolor{myblue}{HTML}{7eb0d5}
\definecolor{mygreen}{HTML}{b2e061}
\definecolor{mypurple}{HTML}{bd7ebe}
\definecolor{myorange}{HTML}{ffb55a}
\definecolor{myyellow}{HTML}{ffee65}
\definecolor{mypurple2}{HTML}{beb9db}
\definecolor{mypink}{HTML}{fdcce5}
\definecolor{mycyan}{HTML}{8bd3c7}
\definecolor{mycyan2}{HTML}{00ffff}
\definecolor{myblue2}{HTML}{115f9a}
\definecolor{myred2}{HTML}{c23728}
\begin{figure}[!t]
\centering
\begin{small}
\resizebox{0.22\textwidth}{!}{
\begin{tikzpicture}
    \begin{customlegend}[legend columns=3,
        legend entries={MiF1,Acc,time},
        legend columns=4,
        legend style={at={(0.45,1.35)},anchor=north,draw=none,font=\small,column sep=0.15cm}]
    \addlegendimage{line width=0.3mm,mark size=3pt,mark=square,color=teal}
    \addlegendimage{line width=0.3mm,mark size=3pt,mark=triangle,color=orange}
    \addlegendimage{line width=0.3mm,mark size=3pt,mark=o,color=purple}
    \end{customlegend}
\end{tikzpicture}}
\\[-\lineskip]
\vspace{-4mm}
\subfloat[Cora-CA]{
\resizebox{0.19\textwidth}{!}{
\begin{tikzpicture}[scale=0.93,every mark/.append style={mark size=2pt}]
    \begin{axis}[
    height=\columnwidth/2.4,
    width=\columnwidth/1.8,
        tick align=center,
        xmin=0.5, xmax=6.5,
        ymin=0.52, ymax=0.88,
        xtick={1,2,3,4,5,6},
        xtick pos=left,
        xticklabel style = {font=\normalsize},
        xticklabels={2,5,10,20,50,100},
        ytick={0.6,0.7,0.8},
        ytick pos=left,
        xlabel = {$K$},
        every axis x label/.style={font=\normalsize,at={{(0.68,0.0)}},right=12mm,below=0.3mm},
        scaled y ticks = false,
        every axis y label/.style={at={(current axis.north west)},right=5mm,above=0mm},
        label style={font=\normalsize},
        tick label style={font=\normalsize},
    ]
    \addplot[line width=0.3mm,mark=square,color=teal]  
    plot coordinates {
    (1,0.720)
    (2,0.737)
    (3,0.753)
    (4,0.744)
    (5,0.741)
    (6,0.722)
    };    
    \addplot[line width=0.3mm,mark=triangle,color=orange] 
    plot coordinates {
    (1,0.725)
    (2,0.751)
    (3,0.766)
    (4,0.762)
    (5,0.742)
    (6,0.719)
    };
    \end{axis}

    \begin{axis}[
    height=\columnwidth/2.4,
    width=\columnwidth/1.8,
        axis lines=left,
        xmin=0.5, xmax=6.5,
        ymin=0.012, ymax=0.088,
        xtick = \empty,
        xticklabel style = {font=\normalsize},
        ytick={0.02,0.04,0.06,0.08},
        yticklabels = {0.02,0.04,0.06,0.08},
        scaled y ticks = false,
        ytick pos=right,
        ylabel={time} (s),
        every axis y label/.style={at={(current axis.north west)},right=30mm,above=1mm},
        label style={font=\normalsize},
        tick label style={font=\normalsize},
    ]
    \addplot[line width=0.3mm,mark=o,color=purple]  
    plot coordinates {
    (1,0.0271)
    (2,0.0274)
    (3,0.0276)
    (4,0.0341)
    (5,0.0562)
    (6,0.0706)
    };    
    \end{axis}
\end{tikzpicture}
}}
\hspace{2mm} 
\subfloat[Amazon]{
\resizebox{0.18\textwidth}{!}{
\begin{tikzpicture}[scale=0.93,every mark/.append style={mark size=2pt}]
    \begin{axis}[
    height=\columnwidth/2.4,
    width=\columnwidth/1.8,
        tick align=center,
        xmin=0.5, xmax=6.5,
        ymin=0.4, ymax=1.0,
        xtick={1,2,3,4,5,6},
        xtick pos=left,
        xticklabel style = {font=\normalsize},
        xticklabels={2,5,10,20,50,100},
        ytick={0.5,0.7,0.9},
        ytick pos=left,
        xlabel = {$K$},
        every axis x label/.style={font=\normalsize,at={{(0.68,0.0)}},right=12mm,below=0.3mm},
        scaled y ticks = false,
        every axis y label/.style={at={(current axis.north west)},right=5mm,above=0mm},
        label style={font=\normalsize},
        tick label style={font=\normalsize},
    ]
    \addplot[line width=0.3mm,mark=square,color=teal]  
    plot coordinates {
    (1,0.543)
    (2,0.693)
    (3,0.718)
    (4,0.722)
    (5,0.706)
    (6,0.692)
    };    
    \addplot[line width=0.3mm,mark=triangle,color=orange] 
    plot coordinates {
    (1,0.903)
    (2,0.909)
    (3,0.909)
    (4,0.901)
    (5,0.896)
    (6,0.891)
    };
    \end{axis}

    \begin{axis}[
    height=\columnwidth/2.4,
    width=\columnwidth/1.8,
        axis lines=left,
        xmin=0.5, xmax=6.5,
        ymin=46, ymax=65,
        xtick = \empty,
        xticklabel style = {font=\normalsize},
        ytick = {48,53,58,63},
        yticklabels = {48,53,58,63},
        scaled y ticks = false,
        ytick pos=right,
        ylabel={time} (s),
        every axis y label/.style={at={(current axis.north west)},right=30mm,above=1mm},
        label style={font=\normalsize},
        tick label style={font=\normalsize},
    ]
    \addplot[line width=0.3mm,mark=o,color=purple]  
    plot coordinates {
    (1,50.5280)
    (2,51.5309)
    (3,52.0086)
    (4,52.7105)
    (5,55.4476)
    (6,60.8413)
    };    
    \end{axis}
\end{tikzpicture}
}}
\end{small}
\vspace{-3mm}
\caption{Varying $K$.} \label{fig:param-knn-small}
\vspace{-3mm}
\end{figure}
\stitle{Varying $K$} Figure \ref{fig:param-knn-small} shows the MiF1 for node classification, Acc for hyperedge link prediction, and the time to construct attribute-based hyperedges in 
$\ES_K$ for the Cora-CA and Amazon datasets. As 
$K$ varies from 2 to 100, time costs rise significantly, especially for 
$K>20$. Embedding quality improves notably as 
$K$ increases from 2 to 10, highlighting the importance of incorporating attribute similarity. However, beyond 
$K=10$, the metrics stabilize and then decline, likely due to the inclusion of nodes with dissimilar attributes, which introduces noise. Thus, we set 
parameter $K$ to 10 over all datasets.

\eat{Figure \ref{fig:param-knn-small} illustrates the micro-F1 scores for node classification task (MiF1), the accuracy score for the hyperedge link prediction task (Acc), and the time for constructing attribute-based hyperedges in $\ES_K$ for Cora-CA and Amazon datasets. We vary $K$ from 2 to 100 and observe that the time costs increase considerably as $K$ rises, particularly when $K$ exceeds 20. Both metrics of embedding quality improve considerably as $K$ increases from 2 to 10, showing the significance of extending hypergraphs with attribute similarity relations. Beyond $K$ values of 10, the metrics start to stabilize and eventually decline. The reason is that with a larger $K$, attribute-based hyperedges contain more nodes with dissimilar attributes, which can introduce noises and undermine the quality of attribute similarities captured. We conclude that setting $K$ to 10 is acceptable for the tested datasets.}

\stitle{Varying $\beta$} Parameter $\beta$  balances  attribute-based hyperedges $\ES_K$ and original hyperedges $\ES$ in $\HGA$.   Figure \ref{fig:param-beta} shows that as $\beta$ increases from 0.1 to 1, micro-F1  for node embedding generally increases, and stabilizes beyond 1.0 on most datasets, but declines for  Cora-CC and MAG-PM when $\beta$ reaches 10.0. The accuracy of hyperedge link prediction (Acc) increases on 20News and Amazon when $\beta$ varies from 0.1 to 1.0, and then remains stable. We set $\beta = 1$ by default.

\definecolor{darkblue}{RGB}{0,0,139}
\definecolor{darkgray176}{RGB}{176,176,176}
\definecolor{darkgreen}{RGB}{0,100,0}
\definecolor{darkorange25512714}{RGB}{255,127,14}
\definecolor{darkred}{RGB}{139,0,0}
\definecolor{forestgreen4416044}{RGB}{44,160,44}
\definecolor{goldenrod}{RGB}{218,165,32}
\definecolor{lightgray204}{RGB}{204,204,204}
\definecolor{lightseagreen0171189}{RGB}{0,171,189}
\definecolor{lightsteelblue161199224}{RGB}{161,199,224}
\definecolor{magenta}{RGB}{255,0,255}
\definecolor{steelblue31119180}{RGB}{31,119,180}
\definecolor{teal2110129}{RGB}{2,110,129}
\definecolor{myred}{HTML}{fd7f6f}
\definecolor{mywhite}{HTML}{D8D8D8}
\definecolor{myorange}{HTML}{D7191C}
\definecolor{myblue}{HTML}{7eb0d5}
\definecolor{mygreen}{HTML}{b2e061}
\definecolor{mypurple}{HTML}{bd7ebe}
\definecolor{myorange}{HTML}{ffb55a}
\definecolor{myyellow}{HTML}{ffee65}
\definecolor{mypurple2}{HTML}{beb9db}
\definecolor{mypink}{HTML}{fdcce5}
\definecolor{mycyan}{HTML}{8bd3c7}
\definecolor{mycyan2}{HTML}{00ffff}
\definecolor{myblue2}{HTML}{115f9a}
\definecolor{myred2}{HTML}{c23728}
\begin{figure}[!t]
\centering
\begin{small}
\resizebox{0.44\textwidth}{!}{
\begin{tikzpicture}
    \begin{customlegend}[
        legend entries={DBLP-CA, Cora-CA, Cora-CC, Citeseer, Mushroom, 20News, DBLP, Recipe, Amazon, MAG-PM},
        legend columns=5,
        legend style={at={(0.05,1.05)},anchor=north,draw=none,font=\small, column sep=0.05cm,row sep=0.05cm}]
    \addlegendimage{line width=0.3mm,mark size=3pt,mark=square,color=teal}
    \addlegendimage{line width=0.3mm,mark size=3pt,mark=triangle,color=orange}
    \addlegendimage{line width=0.3mm,mark size=3pt,mark=o,color=mypurple}
    \addlegendimage{line width=0.3mm,mark size=3pt,mark=diamond,color=magenta}

    \addlegendimage{line width=0.3mm,mark size=3pt,mark=Mercedes star,color=black}
    
    \addlegendimage{line width=0.3mm,mark size=3pt,mark=x,color=myblue2}
    \addlegendimage{line width=0.3mm,mark size=3pt,mark=|,color=mycyan2}

    \addlegendimage{line width=0.3mm,mark size=3pt,mark=Mercedes star flipped,color=goldenrod}
        
    \addlegendimage{line width=0.3mm,mark size=3pt,mark=pentagon,color=darkgray176}
    \addlegendimage{line width=0.3mm,mark size=3pt,mark=10-pointed star,color=myred2}
    \end{customlegend}
\end{tikzpicture}}
\vspace{-4mm}

\subfloat{
\resizebox{0.18\textwidth}{!}{
\begin{tikzpicture}[scale=0.98,every mark/.append style={mark size=2pt}]
    \begin{axis}[
    height=\columnwidth/2.6,
    width=\columnwidth/1.8,
        axis lines=left,
        xmin=0.5, xmax=5.5,
        ymin=0.52, ymax=1.03,
        xtick={1,2,3,4,5},
        xticklabel style = {font=\normalsize},
        xticklabels={0.1,0.5,1.0,2.0,10.0},
        ytick={0.6,0.7,0.8,0.9,1.0},
        scaled y ticks = false,
        ylabel={MiF1},
        every axis y label/.style={at={(current axis.north west)},right=-3mm,above=1mm},
        label style={font=\normalsize},
        tick label style={font=\normalsize},
        xlabel={$\beta$},
        every axis x label/.style={font=\normalsize,at={{(0.68,0.0)}},right=12mm,below=1mm},
    ]
    \addplot[line width=0.3mm,mark=square,color=teal]  
    plot coordinates {
    (1,0.817)
    (2,0.822)
    (3,0.824)
    (4,0.827)
    (5,0.828)
    };    
    \addplot[line width=0.3mm,mark=triangle,color=orange] 
    plot coordinates {
    (1,0.716)
    (2,0.744)
    (3,0.753)
    (4,0.757)
    (5,0.764)
    };
    \addplot[line width=0.3mm,mark=o,color=mypurple] 
    plot coordinates {
    (1,0.719)
    (2,0.737)
    (3,0.742)
    (4,0.732)
    (5,0.730)
    };
    \addplot[line width=0.3mm,mark=diamond,color=magenta] 
    plot coordinates {
    (1,0.673)
    (2,0.678)
    (3,0.690)
    (4,0.682)
    (5,0.683)
    };
    \addplot[line width=0.3mm,mark=x,color=myblue2] 
    plot coordinates {
    (1,0.773)
    (2,0.783)
    (3,0.786)
    (4,0.785)
    (5,0.779)
    };
    \addplot[line width=0.3mm,mark=|,color=mycyan2] 
    plot coordinates {
    (1,0.748)
    (2,0.860)
    (3,0.867)
    (4,0.879)
    (5,0.884)
    };
    \addplot[line width=0.3mm,mark=pentagon,color=darkgray176] 
    plot coordinates {
    (1,0.637)
    (2,0.685)
    (3,0.718)
    (4,0.725)
    (5,0.717)
    };
    \addplot[line width=0.3mm,mark=10-pointed star,color=myred2] 
    plot coordinates {
    (1,0.569)
    (2,0.665)
    (3,0.698)
    (4,0.683)
    (5,0.670)
    };

    \addplot[line width=0.3mm,mark=Mercedes star,color=black] 
    plot coordinates {
    (1,0.997)
    (2,0.999)
    (3,0.999)
    (4,0.999)
    (5,0.999)
    };

    \addplot[line width=0.3mm,mark=Mercedes star flipped,color=goldenrod] 
    plot coordinates {
    (1,0.683)
    (2,0.630)
    (3,0.630)
    (4,0.622)
    (5,0.608)
    };
    
    \end{axis}
\end{tikzpicture}
}}
\subfloat{
\resizebox{0.18\textwidth}{!}{
\begin{tikzpicture}[scale=0.98,every mark/.append style={mark size=2pt}]
    \begin{axis}[
    height=\columnwidth/2.6,
    width=\columnwidth/1.8,
        axis lines=left,
        xmin=0.5, xmax=5.5,
        ymin=0.55, ymax=1.03,
        xtick={1,2,3,4,5},
        xticklabel style = {font=\normalsize},
        xticklabels={0.1,0.5,1.0,2.0,10.0},
        ytick={0.6,0.7,0.8,0.9,1.0},
        scaled y ticks = false,
        ylabel={Acc},
        every axis y label/.style={at={(current axis.north west)},right=-3mm,above=1mm},
        label style={font=\normalsize},
        tick label style={font=\normalsize},
        xlabel={$\beta$},
        every axis x label/.style={font=\normalsize,at={{(0.68,0.0)}},right=12mm,below=1mm},
    ]
    \addplot[line width=0.3mm,mark=square,color=teal] 
    plot coordinates {
    (1,0.748)
    (2,0.766)
    (3,0.776)
    (4,0.781)
    (5,0.784)
    };    
    \addplot[line width=0.3mm,mark=triangle,color=orange] 
    plot coordinates {
    (1,0.750)
    (2,0.753)
    (3,0.766)
    (4,0.759)
    (5,0.759)
    };
    \addplot[line width=0.3mm,mark=o,color=mypurple] 
    plot coordinates {
    (1, 0.803)
    (2, 0.813)
    (3, 0.807)
    (4, 0.807)
    (5, 0.803)
    };
    \addplot[line width=0.3mm,mark=diamond,color=magenta] 
    plot coordinates {
    (1, 0.792)
    (2, 0.809)
    (3, 0.801)
    (4, 0.807)
    (5, 0.806)
    };
    \addplot[line width=0.3mm,mark=x,color=myblue2] 
    plot coordinates {
    (1, 0.605)
    (2, 0.710)
    (3, 0.870)
    (4, 0.880)
    (5, 0.870)
    };
    \addplot[line width=0.3mm,mark=|,color=mycyan2] 
    plot coordinates {
    (1, 0.775)
    (2, 0.821)
    (3, 0.824)
    (4, 0.828)
    (5, 0.829)
    };
    \addplot[line width=0.3mm,mark=pentagon,color=darkgray176] 
    plot coordinates {
    (1,0.719)
    (2,0.906)
    (3,0.909)
    (4,0.905)
    (5,0.902)
    };
    \addplot[line width=0.3mm,mark=10-pointed star,color=myred2] 
    plot coordinates {
    (1,0.719)
    (2,0.769)
    (3,0.761)
    (4,0.755)
    (5,0.743)
    };
    
    \addplot[line width=0.3mm,mark=Mercedes star,color=black] 
    plot coordinates {
    (1,0.982)
    (2,0.985)
    (3,0.989)
    (4,0.987)
    (5,0.987)
    };

    \addplot[line width=0.3mm,mark=Mercedes star flipped,color=goldenrod] 
    plot coordinates {
    (1,0.773)
    (2,0.763)
    (3,0.763)
    (4,0.758)
    (5,0.758)
    };
    
    \end{axis}
\end{tikzpicture}
}}
\end{small}

\vspace{-4mm}
\caption{Varying $\beta$.} \label{fig:param-beta}
\vspace{-7mm}
\end{figure}
\definecolor{darkblue}{RGB}{0,0,139}
\definecolor{darkgray176}{RGB}{176,176,176}
\definecolor{darkgreen}{RGB}{0,100,0}
\definecolor{darkorange25512714}{RGB}{255,127,14}
\definecolor{darkred}{RGB}{139,0,0}
\definecolor{forestgreen4416044}{RGB}{44,160,44}
\definecolor{goldenrod}{RGB}{218,165,32}
\definecolor{lightgray204}{RGB}{204,204,204}
\definecolor{lightseagreen0171189}{RGB}{0,171,189}
\definecolor{lightsteelblue161199224}{RGB}{161,199,224}
\definecolor{magenta}{RGB}{255,0,255}
\definecolor{steelblue31119180}{RGB}{31,119,180}
\definecolor{teal2110129}{RGB}{2,110,129}
\definecolor{myred}{HTML}{fd7f6f}
\definecolor{mywhite}{HTML}{D8D8D8}
\definecolor{myorange}{HTML}{D7191C}
\definecolor{myblue}{HTML}{7eb0d5}
\definecolor{mygreen}{HTML}{b2e061}
\definecolor{mypurple}{HTML}{bd7ebe}
\definecolor{myorange}{HTML}{ffb55a}
\definecolor{myyellow}{HTML}{ffee65}
\definecolor{mypurple2}{HTML}{beb9db}
\definecolor{mypink}{HTML}{fdcce5}
\definecolor{mycyan}{HTML}{8bd3c7}
\definecolor{mycyan2}{HTML}{00ffff}
\definecolor{myblue2}{HTML}{115f9a}
\definecolor{myred2}{HTML}{c23728}
\begin{figure}[!t]
\centering
\begin{small}

\subfloat{
\resizebox{0.18\textwidth}{!}{
\begin{tikzpicture}[scale=0.98,every mark/.append style={mark size=2pt}]
    \begin{axis}[
    height=\columnwidth/2.6,
    width=\columnwidth/1.8,
        axis lines=left,
        xmin=0.5, xmax=5.5,
        ymin=0.475, ymax=1.03,
        xtick={1,2,3,4,5},
        xticklabel style = {font=\normalsize},
        xticklabels={16,24,32,64,128},
        ytick={0.5,0.6,0.7,0.8,0.9,1.0},
        scaled y ticks = false,
        ylabel={MiF1},
        every axis y label/.style={at={(current axis.north west)},right=-3mm,above=1mm},
        label style={font=\normalsize},
        tick label style={font=\normalsize},
        xlabel={$r$},
        every axis x label/.style={font=\normalsize,at={{(0.68,0.0)}},right=12mm,below=1mm},
    ]
    \addplot[line width=0.3mm,mark=square,color=teal]  
    plot coordinates {
    (1,0.827)
    (2,0.826)
    (3,0.824)
    (4,0.711)
    (5,0.617)
    };    
    \addplot[line width=0.3mm,mark=triangle,color=orange] 
    plot coordinates {
    (1,0.728)
    (2,0.746)
    (3,0.753)
    (4,0.743)
    (5,0.518)
    };
    \addplot[line width=0.3mm,mark=o,color=mypurple] 
    plot coordinates {
    (1,0.721)
    (2,0.732)
    (3,0.742)
    (4,0.726)
    (5,0.578)
    };
    \addplot[line width=0.3mm,mark=diamond,color=magenta] 
    plot coordinates {
    (1,0.675)
    (2,0.677)
    (3,0.690)
    (4,0.672)
    (5,0.567)
    };
    \addplot[line width=0.3mm,mark=x,color=myblue2] 
    plot coordinates {
    (1,0.740)
    (2,0.767)
    (3,0.786)
    (4,0.776)
    (5,0.747)
    };
    \addplot[line width=0.3mm,mark=|,color=mycyan2] 
    plot coordinates {
    (1,0.865)
    (2,0.872)
    (3,0.867)
    (4,0.877)
    (5,0.827)
    };
    \addplot[line width=0.3mm,mark=pentagon,color=darkgray176] 
    plot coordinates {
    (1,0.676)
    (2,0.689)
    (3,0.701)
    (4,0.693)
    (5,0.682)
    };
    \addplot[line width=0.3mm,mark=10-pointed star,color=myred2] 
    plot coordinates {
    (1,0.633)
    (2,0.677)
    (3,0.698)
    (4,0.710)
    (5,0.706)
    };
    \addplot[line width=0.3mm,mark=Mercedes star,color=black] 
    plot coordinates {
    (1,0.999)
    (2,0.999)
    (3,1.000)
    (4,0.999)
    (5,0.960)
    };

    \addplot[line width=0.3mm,mark=Mercedes star flipped,color=goldenrod] 
    plot coordinates {
    (1,0.598)
    (2,0.613)
    (3,0.630)
    (4,0.649)
    (5,0.644)
    };
    \end{axis}
\end{tikzpicture}
}}
\subfloat{
\resizebox{0.18\textwidth}{!}{
\begin{tikzpicture}[scale=0.98,every mark/.append style={mark size=2pt}]
    \begin{axis}[
    height=\columnwidth/2.6,
    width=\columnwidth/1.8,
        axis lines=left,
        xmin=0.5, xmax=5.5,
        ymin=0.595, ymax=1.03,
        xtick={1,2,3,4,5},
        xticklabel style = {font=\normalsize},
        xticklabels={16,24,32,64,128},
        ytick={0.6,0.7,0.8,0.9,1.0},
        scaled y ticks = false,
        ylabel={Acc},
        every axis y label/.style={at={(current axis.north west)},right=-3mm,above=1mm},
        label style={font=\normalsize},
        tick label style={font=\normalsize},
        xlabel={$r$},
        every axis x label/.style={font=\normalsize,at={{(0.68,0.0)}},right=12mm,below=1mm},
    ]
    \addplot[line width=0.3mm,mark=square,color=teal] 
    plot coordinates {
    (1,0.775)
    (2,0.785)
    (3,0.776)
    (4,0.754)
    (5,0.746)
    };    
    \addplot[line width=0.3mm,mark=triangle,color=orange] 
    plot coordinates {
    (1,0.750)
    (2,0.759)
    (3,0.766)
    (4,0.723)
    (5,0.693)
    };
    \addplot[line width=0.3mm,mark=o,color=mypurple] 
    plot coordinates {
    (1, 0.803)
    (2, 0.812)
    (3, 0.807)
    (4, 0.785)
    (5, 0.768)
    };
    \addplot[line width=0.3mm,mark=diamond,color=magenta] 
    plot coordinates {
    (1, 0.798)
    (2, 0.806)
    (3, 0.801)
    (4, 0.773)
    (5, 0.773)
    };
    \addplot[line width=0.3mm,mark=x,color=myblue2] 
    plot coordinates {
    (1, 0.860)
    (2, 0.870)
    (3, 0.870)
    (4, 0.625)
    (5, 0.635)
    };
    \addplot[line width=0.3mm,mark=|,color=mycyan2] 
    plot coordinates {
    (1, 0.808)
    (2, 0.823)
    (3, 0.824)
    (4, 0.819)
    (5, 0.788)
    };
    \addplot[line width=0.3mm,mark=pentagon,color=darkgray176] 
    plot coordinates {
    (1,0.895)
    (2,0.907)
    (3,0.909)
    (4,0.905)
    (5,0.898)
    };
    \addplot[line width=0.3mm,mark=10-pointed star,color=myred2] 
    plot coordinates {
    (1,0.699)
    (2,0.749)
    (3,0.761)
    (4,0.765)
    (5,0.753)
    };

    \addplot[line width=0.3mm,mark=Mercedes star,color=black] 
    plot coordinates {
    (1,0.983)
    (2,0.983)
    (3,0.948)
    (4,0.953)
    (5,0.956)
    };

    \addplot[line width=0.3mm,mark=Mercedes star flipped,color=goldenrod] 
    plot coordinates {
    (1,0.769)
    (2,0.765)
    (3,0.763)
    (4,0.753)
    (5,0.724)
    };
\end{axis}
\end{tikzpicture}
}}

\end{small}

\vspace{-4mm}
\caption{Varying $r$.} \label{fig:param-rank}
\submission{\vspace{-5mm}}
\report{\vspace{-2mm}}
\end{figure}

\report{

\report{\begin{table*}[!t] 
\centering 
\caption{Ablation analysis of \nprox on node classification performance.}\vspace{-3mm}
 
\resizebox{0.86\textwidth}{!}{
\renewcommand{\arraystretch}{0.95}
\setlength{\tabcolsep}{6pt}
\begin{tabular}{|c|cc|cc|cc|cc|cc|cc|cc|}
\hline
\multirow{2}{*}{\bf{Method}} & \multicolumn{2}{c|}{\bf{DBLP-CA}} & \multicolumn{2}{c|}{\bf{Cora-CA}} & \multicolumn{2}{c|}{\bf{Cora-CC}} & \multicolumn{2}{c|}{\bf{Citeseer}}  & \multicolumn{2}{c|}{\bf{Mushroom}} & \multicolumn{2}{c|}{\bf{20News}} & \multicolumn{2}{c|}{\bf{DBLP}}\\ \cline{2-15}

&MiF1 & MaF1 &MiF1 & MaF1 &MiF1 & MaF1 &MiF1 & MaF1 &MiF1 & MaF1 &MiF1 & MaF1 &MiF1 & MaF1 \\ \hline

\nprox-no-$\ESK$ &0.527	&0.492 &0.507	&0.462 &0.559	&0.530 &0.329	&0.278 &0.990	&0.990 &{0.796} &{0.769} &0.755	&0.733 \\
\nprox-1-hop &0.811 &0.803 &0.741	&0.720 &0.716	&0.696 &0.680	&0.618 &0.988	&0.988 &0.783	&0.755 &0.852	&0.843 \\
{\nprox} & \textbf{0.836} & \textbf{0.828} & \textbf{0.777} & \textbf{0.754} & \textbf{0.753} & \textbf{0.732} & \textbf{0.693} & \textbf{0.628} &\textbf{0.997}	&\textbf{0.997}  & \textbf{0.801} & \textbf{0.775} & \textbf{0.898} & \textbf{0.894} \\
\hline
\end{tabular}
}
\label{tab:ablation-node-classification}
\vspace{-2mm}
\end{table*}}

\eat{
\begin{table*}[!t]
\centering
\caption{Ablation analysis results on hyperedge link prediction performance.}\vspace{-3mm}

\resizebox{0.98\textwidth}{!}{
\renewcommand{\arraystretch}{0.95}
\setlength{\tabcolsep}{3pt}
\begin{tabular}{|c|cc|cc|cc|cc|cc|cc|cc|}
\hline
\multirow{2}{*}{\bf{Method}} & \multicolumn{2}{c|}{\bf{DBLP-CA}} & \multicolumn{2}{c|}{\bf{Cora-CA}} & \multicolumn{2}{c|}{\bf{Cora-CC}} & \multicolumn{2}{c|}{\bf{Citeseer}}  & \multicolumn{2}{c|}{\bf{Mushroom}} & \multicolumn{2}{c|}{\bf{20News}} & \multicolumn{2}{c|}{\bf{DBLP}} \\ \cline{2-15}

&Acc & AUC &Acc & AUC  &Acc & AUC  &Acc & AUC  &Acc & AUC  &Acc & AUC  &Acc & AUC  \\ \hline

\nprox-1-hop &\first{0.829}	&\first{0.904} &\first{0.732}	&0.799 &0.759	&0.837 &0.800	&0.887 &0.960	&0.994 &0.740	&0.911 &\first{0.805}	&0.881\\
\nprox-no-$\ESK$  &0.670	&0.716 &0.731	&0.789 &0.714	&0.734 &0.634	&0.639 &0.962	&0.994 &0.495 &0.411 &0.754	&0.803 \\

\nprox     & {0.788} & {0.894} & {0.727} & \first{0.817} & \first{0.834} & \first{0.931} & \first{0.819} & \first{0.922} &\first{0.968}	&\first{0.996} & \first{0.955} & \first{0.982} & \first{0.805} & \first{0.898} \\
\hline
\end{tabular}
}
\label{tab:ablation-link-prediction}
\vspace{-0mm}
\end{table*}}

\begin{table}[!t]
\centering
\submission{\caption{\revision{Ablation analysis.}}\vspace{-4mm}}
\report{\caption{\revision{Ablation analysis of \heprox on hyperedge classification performance.}}\vspace{-4mm}}
\hspace{-3mm} 
\resizebox{1.02\linewidth}{!}{
\renewcommand{\arraystretch}{0.98}
\setlength{\tabcolsep}{1.5pt} 
\begin{tabular}{|c|cc|cc|cc|cc|cc|cc|} 
\hline
\multirow{2}{*}{\bf{Method}} & \multicolumn{2}{c|}{\bf{DBLP-CA}} & \multicolumn{2}{c|}{\bf{Cora-CA}} & \multicolumn{2}{c|}{\bf{Cora-CC}} & \multicolumn{2}{c|}{\bf{Citeseer}}  & \multicolumn{2}{c|}{\bf{Mushroom}}  & \multicolumn{2}{c|}{\bf{DBLP}}\\ \cline{2-13}

&MiF1 & MaF1 &MiF1 & MaF1 &MiF1 & MaF1 &MiF1 & MaF1 &MiF1 & MaF1 &MiF1 & MaF1 \\ \hline

\heprox-1-hop &0.648 & 0.605 &0.489 & 0.423 &0.765 &0.749 &0.610 &0.521 & 0.899 & 0.897 & 0.650 &0.552\\
\heprox-no-$\ESK$ &0.658 &0.615 &0.487 & 0.418 &0.821 &0.810 & 0.623 &0.549 & 0.921 & 0.920 & 0.747 & 0.716\\
{\heprox} &\textbf{0.858} &\textbf{0.838} &\textbf{0.775} &\textbf{0.740} &\textbf{0.852} &\textbf{0.846} &\textbf{0.770} &\textbf{0.684} &\textbf{0.926} &\textbf{0.925} &\textbf{0.908} &\textbf{0.898} \\
\hline
\end{tabular}
}
\label{tab:ablation-hyperedge-classification}
\vspace{-3mm}
\end{table}
}

\stitle{Varying $r$} The parameter $r$ represents the dimension of truncated SVD used for approximating \nprox and \heprox in Section \ref{sec:approxNonliear}. We vary $r$ for node classification (MiF1) and hyperedge link prediction (Acc), with results in Figure \ref{fig:param-rank}.
\revision{Increasing $r$ from 16 to 32 generally improves or stabilizes both metrics, except for Mushroom, where Acc drops after 24. Beyond 32, scores typically decline. Thus, we set $r=32$ by default and $r=16$ for the Mushroom dataset.}

\begin{figure}[!t]

\centering

\captionsetup[subfloat]{labelfont={small}, textfont={small}}
\subfloat{
\resizebox{0.33\linewidth}{!}{%
\includegraphics[width=0.15\textwidth]{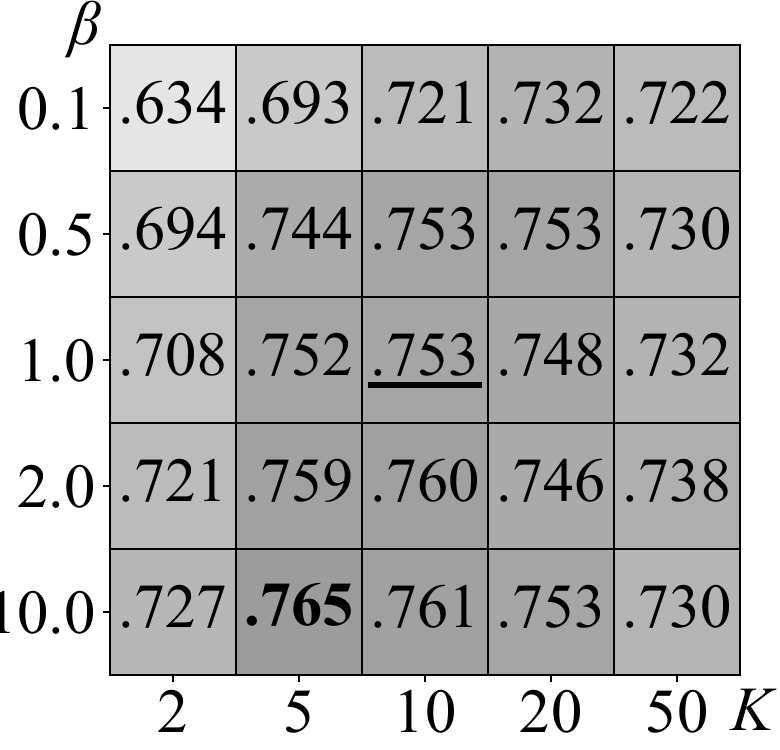}
 }}
 \captionsetup[subfloat]{labelfont={small}, textfont={small}}\hspace{-2mm}
\subfloat{
\resizebox{0.33\linewidth}{!}{%
\includegraphics[width=0.15\textwidth]{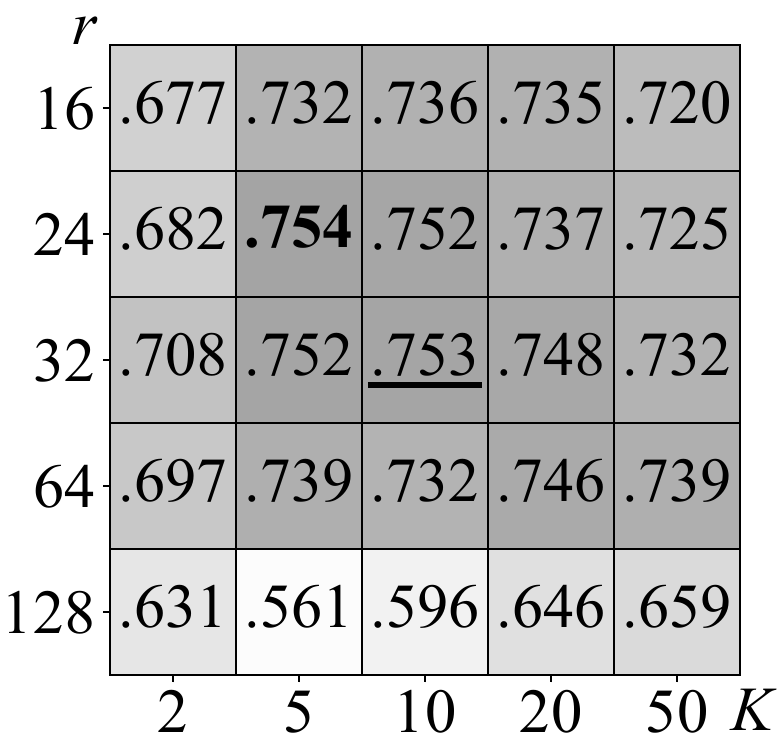}
 }}
\captionsetup[subfloat]{labelfont={small}, textfont={small}}\hspace{-2mm}
\subfloat{
\resizebox{0.33\linewidth}{!}{%
\includegraphics[width=0.15\textwidth]{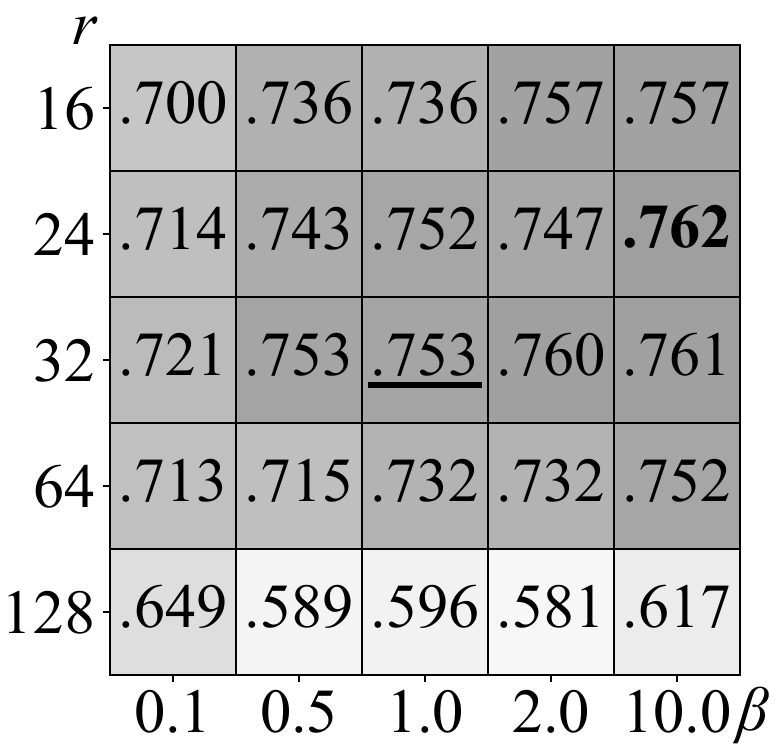}
 }}
\vspace{-3mm}
\caption{\revision{Heatmaps between $K$, $\beta$, and $r$ on Cora-CA for node classification. Darker shades indicate higher MiF1. \underline{Underlined} results are reported in \hardcode{Table~\ref{tab:nc_performance}}, {while the optimal parameter combinations are in bold}.}}
\label{fig:param-interaction} 
\vspace{-3mm}
\end{figure}
\revision{
\stitle{Parameter interaction} We analyze the interaction between parameters $K$, $\beta$, and $r$ using node classification on Cora-CA, shown in Figure~\ref{fig:param-interaction}. The heatmaps display Micro-F1 scores across parameter combinations. \ours consistently delivers high-quality embeddings, indicated by darker gray levels, confirming robustness to parameter variations. Underlined results are those reported in Table~\ref{tab:nc_performance}, while bold values represent optimal performance with fine-tuned parameters, surpassing defaults. 
This shows \ours delivers strong performance with default settings,  without extensive tuning, and provides guidance for adjusting parameters in practice.
\submission{More parameter analysis is in the technical report~\cite{report}.}}

\submission{
\revision{
\stitle{Ablation study}
We conduct an ablation study on the proposed \heprox measure by evaluating two variants: \heprox-no-$\ESK$, which excludes attribute-based hyperedges, and \heprox-1-hop, which restricts random walks to a single hop. Hyperedge embeddings derived from these versions are assessed via the classification task, with results in Table~\ref{tab:ablation-hyperedge-classification}, where the full \heprox outperforms both ablated versions, confirming the effectiveness of our approach. The ablation for \nprox is similar and reported in~\cite{report}.}
}

\report{
\revision{
\stitle{Ablation study} To validate our proposed similarity measure formulation, we evaluate two ablated versions: \nprox/\heprox-no-$\ESK$, which excludes attribute-based hyperedges, and \nprox/\heprox-1-hop, which restricts random walks to a single hop. Node and hyperedge embeddings derived from these similarity matrices are assessed via the classification task, with results in Tables~\ref{tab:ablation-node-classification}-\ref{tab:ablation-hyperedge-classification}. The full \nprox and \heprox measures generally outperform both ablated versions, confirming the effectiveness of our approach.
}}

\report{\input{figure_revision/new_baseline}}

\report{
\revision{
\stitle{Adapting \nprox and \heprox to existing methods} We explore two strategies to integrate the key intuitions of \nprox and \heprox into existing methods, evaluating their impact on embedding quality alongside \ours.
First, we adapt our multi-hop random walk process into the \hypertovec framework, yielding \hypertovecplus for node and hyperedge embeddings. Second, for graph neural network models that do not use random walks, such as \tricl, we enhance it by concatenating the \nprox matrix with node features, resulting in \triclplus, where hyperedge embeddings are derived by averaging node embeddings. 
Tables \ref{tab:added-baselines-node-classification}-\ref{tab:added-baselines-hyperedge-classification} summarize the results for node classification, hyperedge link prediction, and hyperedge classification tasks, respectively. The last column shows the average rank of each method across all datasets, with lower ranks indicating better overall performance. 
The extended baselines benefit from incorporating our similarity measures to varying extents. For instance, \hypertovecplus outperforms \hypertovec, while \triclplus shows occasional improvements over \tricl.  Besides, \hypertovecplus integrates our random walk scheme, making it more efficient than \hypertovec's second-order random walks, enabling it to process datasets like 20News and Mushroom within the 24-hour limit.
These results highlight that the impact of our ideas depends on the baseline design. Overall, \ours achieves the best overall rank, consistently delivering superior performance across diverse settings and scaling to large datasets that other methods cannot handle.
}}

\report{
\input{figure_revision/varyparam}}

\report{
\revision{
\stitle{Varying $\alpha$} Figure \ref{fig:param-alpha} shows MiF1 for node classification in (a) and Acc for hyperedge link prediction in (b) as $\alpha$ varies from 0.001 to 0.3. \ours demonstrates consistent performance across most values, except for the very small $\alpha=0.001$. This confirms \ours's robustness to different $\alpha$ values, with $\alpha=0.1$ chosen as the default setting.

\stitle{Varying $\tau$ and $b$} Parameters $\tau$ and $b$ balance approximation accuracy and efficiency, with larger values improving accuracy at higher computational cost. Figure~\ref{fig:param-degree}(a) and (b) show MiF1 for node classification and running time as $\tau$ varies from 1 to 9, while Figure~\ref{fig:param-skedim} depicts results for $b$ ranging from 16 to 256. MiF1 initially improves and then stabilizes as these parameters increase, while running time continues to rise. The default settings of $\tau=3$ and $b=128$ effectively balance quality and efficiency, and performance remains robust across a range of values.}}

\section{Related work}\label{sec:relatedwork}

Existing methods struggle to support attributed hypergraphs natively while scaling for massive data, particularly for the \AHE problem that embeds both nodes and hyperedges. Early hypergraph embedding efforts, like~\cite{zhouLearningHypergraphsClustering2007}, use the Laplacian matrix spectrum for node embeddings, focusing on clustering but neglecting attributes and long-range connectivity. Methods extending node2vec~\cite{groverNode2vecScalableFeature2016} to hypergraphs, such as \texttt{Hyper2vec}~\cite{huangHyper2vecBiasedRandom2019} and its dual-enhanced version~\cite{huangNonuniformHyperNetworkEmbedding2020}, capture long-range relationships via random walks, yet omit attributes and hyperedge embeddings. Another approach ~\cite{yuModelingMultiwayRelations2018a} models hyperedges as multi-linear products of node embeddings, but lacks attribute consideration and scalability on large datasets.

\submission{
}

The emergence of hypergraph neural networks also enables a multitude of recent approaches. \texttt{TriCL}~\cite{leeImMeWere2023} uses contrastive learning to embed nodes, \texttt{HypeBoy}~\cite{kimHypeBoyGenerativeSelfSupervised2023} masks attributes and designs a hyperedge-filling task for  node embedding, and~\cite{duSelfsupervisedHypergraphRepresentation2022} applies GNNs on an expanded graph with cluster-based loss. Nevertheless, these methods focus solely on nodes, incur high training costs (e.g., $O(n^2)$ or worse), and lack hyperedge embedding support.
\villain~\cite{leeVilLainSelfSupervisedLearning2024} formulates the self-supervision as a label propagation process while ignoring node attributes and incurring high training costs. There are  specialized methods to learn node representations for certain hypergraphs, like~\cite{yang_lbsn2vec_2022} for location-based networks,~\cite{xiaSelfSupervisedHypergraphTransformer2022} for  recommendations,~\cite{xuAdaptiveHypergraphNetwork2024} for trust relations, and~\cite{liSpatialTemporalHypergraphSelfSupervised2022} for   crime data, but their targeted design for domain-specific data and limited scalability hinder the applicability for general attributed hypergraphs.

Alternatively, embedding techniques designed for graphs or bipartite graphs can be adapted for attributed hypergraphs.
With each hyperedge converted to a fully connected subgraph via clique-expansion, hypergraphs can be processed by graph embedding methods (\texttt{NetMF}~\cite{qiuNetworkEmbeddingMatrix2018a}, \texttt{STRAP}~\cite{yinScalableGraphEmbeddings2019}, \texttt{LightNE}~\cite{qiuLightNELightweightGraph2021}, etc.) or attributed graph embedding method~\cite{yangScalingAttributedNetwork2020} based on matrix factorization. Star expansion of hypergraphs results in dense edges between two sets of nodes, enabling the adoption of bipartite graph embedding methods (\biane~\cite{huangBiANEBipartiteAttributed2020a}, \texttt{AnchorGNN}~\cite{wuBillionScaleBipartiteGraph2023}). However, these transformations weaken higher-order connections critical to \AHE while producing dense graphs with high complexity, resulting in compromised embedding quality and efficiency, as our experiments show. 

\section{Conclusion}\label{sec:conclusion}
This paper proposes the \ours algorithm to tackle the \AHE problem, by efficiently generating node and hyperedge embeddings that preserve the higher-order connectivity among nodes and collective similarities among hyperedges in attributed hypergraphs. By formulating multi-hop similarity measures
and employing optimized decomposition techniques, \ours achieves log-linear time complexity, outperforming 11 baselines across  10 real-world datasets. Its scalability and effectiveness make it well-suited for practical applications. 
Moving forward, we aim to enhance \ours with a dynamic embedding algorithm using incremental decomposition for evolving hypergraphs and a GPU implementation to boost efficiency.

\begin{acks}
This work is supported by grants from the Research Grants Council of Hong Kong Special Administrative Region, China (No. PolyU 25201221, No. PolyU 15205224), and NSFC No. 62202404.
This work is supported by   Smart Cities Research Institute (SCRI) P0051036-P0050643, and grant P0048213 from Tencent Technology Co., Ltd.

\end{acks} 

\report{
\appendix
\input{sections/appendix}
}

\clearpage
\balance
\bibliographystyle{ACM-Reference-Format}
\bibliography{sample}

@String{Computing = "Computing" }

@String{Springer = "Springer-Verlag" }

@ArtifactSoftware{R,
    title = {R: A Language and Environment for Statistical Computing},
    author = {{R Core Team}},
    organization = {R Foundation for Statistical Computing},
    address = {Vienna, Austria},
    year = {2019},
    url = {https://www.R-project.org/},
}

@inproceedings{zhouLearningHypergraphsClustering2007,
  title = {Learning with {{Hypergraphs}}: {{Clustering}}, {{Classification}}, and {{Embedding}}},
  booktitle = {NeurIPS},
  author = {Zhou, Dengyong and Huang, Jiayuan and Sch{\"o}lkopf, Bernhard},
  year = {2007},
  volume = {19}
}

@inproceedings{chitraRandomWalksHypergraphs2019a,
  title = {Random {{Walks}} on {{Hypergraphs}} with {{Edge-Dependent Vertex Weights}}},
  booktitle = {ICML},
  author = {Chitra, Uthsav and Raphael, Benjamin},
  year = {2019},
  pages = {1172--1181},
  publisher = {PMLR}
}

@article{yangScalingAttributedNetwork2020,
  title = {Scaling Attributed Network Embedding to Massive Graphs},
  author = {Yang, Renchi and Shi, Jieming and Xiao, Xiaokui and Yang, Yin and Liu, Juncheng and Bhowmick, Sourav S.},
  year = {2020},
  month = sep,
  journal = {VLDB},
  volume = {14},
  number = {1},
  pages = {37--49}
}

@inproceedings{qiuNetworkEmbeddingMatrix2018a,
  title = {Network {{Embedding}} as {{Matrix Factorization}}: {{Unifying DeepWalk}}, {{LINE}}, {{PTE}}, and Node2vec},
  booktitle = {WSDM},
  author = {Qiu, Jiezhong and Dong, Yuxiao and Ma, Hao and Li, Jian and Wang, Kuansan and Tang, Jie},
  year = {2018},
  pages = {459--467}
}

@inproceedings{hanPolynomialTensorSketch2020,
  title = {Polynomial {{Tensor Sketch}} for {{Element-wise Function}} of {{Low-Rank Matrix}}},
  booktitle = {ICML},
  author = {Han, Insu and Avron, Haim and Shin, Jinwoo},
  year = {2020},
  pages = {3984--3993}
}

@article{lehoucqDeflationTechniquesImplicitly1996,
  title = {Deflation {{Techniques}} for an {{Implicitly Restarted Arnoldi Iteration}}},
  author = {Lehoucq, R. B. and Sorensen, D. C.},
  year = {1996},
  journal = {SIAM Journal on Matrix Analysis and Applications},
  volume = {17},
  number = {4},
  pages = {789--821}
}

@inproceedings{yuModelingMultiwayRelations2018a,
  title = {Modeling {{Multi-way Relations}} with {{Hypergraph Embedding}}},
  booktitle = {CIKM},
  author = {Yu, Chia-An and Tai, Ching-Lun and Chan, Tak-Shing and Yang, Yi-Hsuan},
  year = {2018},
  month = oct,
  series = {{{CIKM}} '18},
  pages = {1707--1710}
}

@inproceedings{huangHyper2vecBiasedRandom2019,
  title = {Hyper2vec: {{Biased Random Walk}} for {{Hyper-network Embedding}}},
  shorttitle = {Hyper2vec},
  booktitle = {DASFAA},
  author = {Huang, Jie and Chen, Chuan and Ye, Fanghua and Wu, Jiajing and Zheng, Zibin and Ling, Guohui},
  year = {2019},
  pages = {273--277}
}

@inproceedings{groverNode2vecScalableFeature2016,
  title = {Node2vec: {{Scalable Feature Learning}} for {{Networks}}},
  shorttitle = {Node2vec},
  booktitle = {KDD},
  author = {Grover, Aditya and Leskovec, Jure},
  year = {2016},
  month = aug,
  pages = {855--864}
}

@article{huangNonuniformHyperNetworkEmbedding2020,
  title = {Nonuniform {{Hyper-Network Embedding}} with {{Dual Mechanism}}},
  author = {Huang, Jie and Chen, Chuan and Ye, Fanghua and Hu, Weibo and Zheng, Zibin},
  year = {2020},
  month = may,
  journal = {ACM Trans. Inf. Syst.},
  volume = {38},
  number = {3},
  pages = {28:1--28:18}
}

@inproceedings{duSelfsupervisedHypergraphRepresentation2022,
  title = {Self-Supervised {{Hypergraph Representation Learning}}},
  booktitle = {IEEE {Big Data}},
  author = {Du, Boxin and Yuan, Changhe and Barton, Robert and Neiman, Tal and Tong, Hanghang},
  year = {2022},
  month = dec,
  pages = {505--514}
}

@article{leeImMeWere2023,
  title = {I'm {{Me}}, {{We}}'re {{Us}}, and {{I}}'m {{Us}}: {{Tri-directional Contrastive Learning}} on {{Hypergraphs}}},
  shorttitle = {I'm {{Me}}, {{We}}'re {{Us}}, and {{I}}'m {{Us}}},
  author = {Lee, Dongjin and Shin, Kijung},
  year = {2023},
  month = jun,
  journal = {AAAI},
  volume = {37},
  number = {7},
  pages = {8456--8464}
}

@inproceedings{kimHypeBoyGenerativeSelfSupervised2023,
  title = {{{HypeBoy}}: {{Generative Self-Supervised Representation Learning}} on {{Hypergraphs}}},
  shorttitle = {{{HypeBoy}}},
  booktitle = {ICLR},
  author = {Kim, Sunwoo and Kang, Shinhwan and Bu, Fanchen and Lee, Soo Yong and Yoo, Jaemin and Shin, Kijung},
  year = {2023},
  month = oct
}

@inproceedings{leeVilLainSelfSupervisedLearning2024,
  title = {{{VilLain}}: {{Self-Supervised Learning}} on {{Homogeneous Hypergraphs}} without {{Features}} via {{Virtual Label Propagation}}},
  shorttitle = {{{VilLain}}},
  booktitle = {WWW},
  author = {Lee, Geon and Lee, Soo Yong and Shin, Kijung},
  year = {2024},
  month = may,
  pages = {594--605}
}

@inproceedings{yinScalableGraphEmbeddings2019,
  title = {Scalable {{Graph Embeddings}} via {{Sparse Transpose Proximities}}},
  booktitle = {KDD},
  author = {Yin, Yuan and Wei, Zhewei},
  year = {2019},
  month = jul,
  series = {{{KDD}} '19},
  pages = {1429--1437}
}

@inproceedings{qiuLightNELightweightGraph2021,
  title = {{{LightNE}}: {{A Lightweight Graph Processing System}} for {{Network Embedding}}},
  shorttitle = {{{LightNE}}},
  booktitle = {SIGMOD},
  author = {Qiu, Jiezhong and Dhulipala, Laxman and Tang, Jie and Peng, Richard and Wang, Chi},
  year = {2021},
  month = jun,
  pages = {2281--2289}
}

@article{yangHomogeneousNetworkEmbedding2020,
  title = {Homogeneous Network Embedding for Massive Graphs via Reweighted Personalized {{PageRank}}},
  author = {Yang, Renchi and Shi, Jieming and Xiao, Xiaokui and Yang, Yin and Bhowmick, Sourav S.},
  year = {2020},
  month = jan,
  journal = {VLDB},
  volume = {13},
  number = {5},
  pages = {670--683}
}

@article{yangPANEScalableEffective2023,
  title = {{{PANE}}: Scalable and Effective Attributed Network Embedding},
  shorttitle = {{{PANE}}},
  author = {Yang, Renchi and Shi, Jieming and Xiao, Xiaokui and Yang, Yin and Bhowmick, Sourav S. and Liu, Juncheng},
  year = {2023},
  month = nov,
  journal = {VLDBJ},
  volume = {32},
  number = {6},
  pages = {1237--1262},
}

@inproceedings{huangBiANEBipartiteAttributed2020a,
  title = {{{BiANE}}: {{Bipartite Attributed Network Embedding}}},
  shorttitle = {{{BiANE}}},
  booktitle = {SIGIR},
  author = {Huang, Wentao and Li, Yuchen and Fang, Yuan and Fan, Ju and Yang, Hongxia},
  year = {2020},
  month = jul,
  series = {{{SIGIR}} '20},
  pages = {149--158}
}

@article{wuBillionScaleBipartiteGraph2023,
  title = {Billion-{{Scale Bipartite Graph Embedding}}: {{A Global-Local Induced Approach}}},
  shorttitle = {Billion-{{Scale Bipartite Graph Embedding}}},
  author = {Wu, Xueyi and Xu, Yuanyuan and Zhang, Wenjie and Zhang, Ying},
  year = {2023},
  month = oct,
  journal = {VLDB},
  volume = {17},
  number = {2},
  pages = {175--183},
}

@article{liEfficientEffectiveAttributed2023,
  title = {Efficient and {{Effective Attributed Hypergraph Clustering}} via {{K-Nearest Neighbor Augmentation}}},
  author = {Li, Yiran and Yang, Renchi and Shi, Jieming},
  year = {2023},
  month = jun,
  journal = {Proc. ACM Manag. Data},
  volume = {1},
  number = {2},
  pages = {116:1--116:23},
}

@article{yangEfficientAlgorithmsPersonalized2024,
  title = {Efficient {{Algorithms}} for {{Personalized PageRank Computation}}: {{A Survey}}},
  shorttitle = {Efficient {{Algorithms}} for {{Personalized PageRank Computation}}},
  author = {Yang, Mingji and Wang, Hanzhi and Wei, Zhewei and Wang, Sibo and Wen, Ji-Rong},
  year = {2024},
  month = sep,
  journal = {TKDE},
  volume = {36},
  number = {9},
  pages = {4582--4602},
  urldate = {2025-01-09}
}

@inproceedings{ameranisFastAlgorithmsHypergraph2024,
  title = {Fast {{Algorithms}} for {{Hypergraph PageRank}} with {{Applications}} to {{Semi-Supervised Learning}}},
  booktitle = {ICML},
  author = {Ameranis, Konstantinos and DePavia, Adela Frances and Orecchia, Lorenzo and Tani, Erasmo},
  year = {2024},
  month = jun,
  urldate = {2024-08-20}
}

@article{takaiHypergraphClusteringBased2020a,
  title = {Hypergraph {{Clustering Based}} on {{PageRank}}},
  author = {Takai, Yuuki and Miyauchi, Atsushi and Ikeda, Masahiro and Yoshida, Yuichi},
  year = {2020},
  month = jun,
  journal = {arXiv:2006.08302 [cs, math]},
  eprint = {2006.08302},
  urldate = {2021-10-04},
  archiveprefix = {arXiv}
}

@article{douze2024faiss,
      title={The Faiss library},
      author={Matthijs Douze and Alexandr Guzhva and Chengqi Deng and Jeff Johnson and Gergely Szilvasy and Pierre-Emmanuel Mazaré and Maria Lomeli and Lucas Hosseini and Hervé Jégou},
      year={2024},
      eprint={2401.08281},
      archivePrefix={arXiv},
      primaryClass={cs.LG}
}

@inproceedings{liu2022robust,
  title={Robust attributed network embedding preserving community information},
  author={Liu, Yunfei and Liu, Zhen and Feng, Xiaodong and Li, Zhongyi},
  booktitle={ICDE},
  pages={1874--1886},
  year={2022},
  organization={IEEE}
}

@article{tan2023collaborative,
  title={Collaborative graph neural networks for attributed network embedding},
  author={Tan, Qiaoyu and Zhang, Xin and Huang, Xiao and Chen, Hao and Li, Jundong and Hu, Xia},
  journal={TKDE},
  year={2023},
  publisher={IEEE}
}

@article{yadati2019hypergcn,
  title={Hypergcn: A new method for training graph convolutional networks on hypergraphs},
  author={Yadati, Naganand and Nimishakavi, Madhav and Yadav, Prateek and Nitin, Vikram and Louis, Anand and Talukdar, Partha},
  journal={NeurIPS},
  volume={32},
  year={2019}
}

@inproceedings{patil2020negative,
  title={Negative sampling for hyperlink prediction in networks},
  author={Patil, Prasanna and Sharma, Govind and Murty, M Narasimha},
  booktitle={PAKDD},
  pages={607--619},
  year={2020},
  organization={Springer}
}

@article{bensonSimplicialClosureHigherorder2018,
  title = {Simplicial Closure and Higher-Order Link Prediction},
  author = {Benson, Austin R. and Abebe, Rediet and Schaub, Michael T. and Jadbabaie, Ali and Kleinberg, Jon},
  year = {2018},
  journal = {PNAS},
  volume = {115},
  number = {48},
  pages = {E11221--E11230}
}

@article{fengHypergraphModelsBiological2021,
  title = {Hypergraph Models of Biological Networks to Identify Genes Critical to Pathogenic Viral Response},
  author = {Feng, Song and Heath, Emily and Jefferson, Brett and Joslyn, Cliff and Kvinge, Henry and Mitchell, Hugh D. and Praggastis, Brenda and Eisfeld, Amie J. and Sims, Amy C. and Thackray, Larissa B. and Fan, Shufang and Walters, Kevin B. and Halfmann, Peter J. and {Westhoff-Smith}, Danielle and Tan, Qing and Menachery, Vineet D. and Sheahan, Timothy P. and Cockrell, Adam S. and Kocher, Jacob F. and Stratton, Kelly G. and Heller, Natalie C. and Bramer, Lisa M. and Diamond, Michael S. and Baric, Ralph S. and Waters, Katrina M. and Kawaoka, Yoshihiro and McDermott, Jason E. and Purvine, Emilie},
  year = {2021},
  month = may,
  journal = {BMC Bioinformatics},
  volume = {22},
  number = {1},
  pages = {287},
  urldate = {2025-02-21}
}

@inproceedings{hanSearchBehaviorPrediction2023,
  title = {Search {{Behavior Prediction}}: {{A Hypergraph Perspective}}},
  shorttitle = {Search {{Behavior Prediction}}},
  booktitle = {WSDM},
  author = {Han, Yan and Huang, Edward W. and Zheng, Wenqing and Rao, Nikhil and Wang, Zhangyang and Subbian, Karthik},
  year = {2023},
  month = feb,
  series = {{{WSDM}} '23},
  pages = {697--705},
}

@inproceedings{liSpatialTemporalHypergraphSelfSupervised2022,
  title = {Spatial-{{Temporal Hypergraph Self-Supervised Learning}} for {{Crime Prediction}}},
  booktitle = {ICDE},
  author = {Li, Zhonghang and Huang, Chao and Xia, Lianghao and Xu, Yong and Pei, Jian},
  year = {2022},
  pages = {2984--2996},
}

@inproceedings{xiaSelfSupervisedHypergraphTransformer2022,
  title = {Self-{{Supervised Hypergraph Transformer}} for {{Recommender Systems}}},
  booktitle = {KDD},
  author = {Xia, Lianghao and Huang, Chao and Zhang, Chuxu},
  year = {2022},
  month = aug,
  pages = {2100--2109},
}

@inproceedings{xuAdaptiveHypergraphNetwork2024,
  title = {Adaptive {{Hypergraph Network}} for {{Trust Prediction}}},
  booktitle = {ICDE},
  author = {Xu, Rongwei and Liu, Guanfeng and Wang, Yan and Zhang, Xuyun and Zheng, Kai and Zhou, Xiaofang},
  year = {2024},
  pages = {2986--2999},
}

@article{tsitsulinFREDEAnytimeGraph2021a,
  title = {{{FREDE}}: Anytime Graph Embeddings},
  shorttitle = {{{FREDE}}},
  author = {Tsitsulin, Anton and Munkhoeva, Marina and Mottin, Davide and Karras, Panagiotis and Oseledets, Ivan and M{\"u}ller, Emmanuel},
  year = {2021},
  month = feb,
  journal = {VLDB},
  volume = {14},
  number = {6},
  pages = {1102--1110},
}

@inproceedings{wuAttributedNetworkEmbedding2024,
  title = {Attributed {{Network Embedding}} in {{Streaming Style}}},
  booktitle = {ICDE},
  author = {Wu, Anbiao and Yuan, Ye and Li, Changsheng and Ma, Yuliang and Zhang, Hao},
  year = {2024},
  month = may,
  pages = {3138--3150},
  urldate = {2024-10-24},
}

@article{yang_lbsn2vec_2022,
    title = {{LBSN2Vec}++: {Heterogeneous} {Hypergraph} {Embedding} for {Location}-{Based} {Social} {Networks}},
    volume = {34},
    shorttitle = {{LBSN2Vec}++},
    number = {4},
    journal = {TKDE},
    author = {Yang, Dingqi and Qu, Bingqing and Yang, Jie and Cudré-Mauroux, Philippe},
    month = apr,
    year = {2022},
    pages = {1843--1855},
}

@article{wei_augmentations_2022,
    title = {Augmentations in {Hypergraph} {Contrastive} {Learning}: {Fabricated} and {Generative}},
    volume = {35},
    shorttitle = {Augmentations in {Hypergraph} {Contrastive} {Learning}},
    language = {en},
    journal = {NeurIPS},
    author = {Wei, Tianxin and You, Yuning and Chen, Tianlong and Shen, Yang and He, Jingrui and Wang, Zhangyang},
    month = dec,
    year = {2022},
    pages = {1909--1922},
}

@article{fengGraphRepresentationAttributed2025,
  title = {On {{Graph Representation}} for {{Attributed Hypergraph Clustering}}},
  author = {Feng, Zijin and Qiao, Miao and Piao, Chengzhi and Cheng, Hong},
  year = {2025},
  month = feb,
  journal = {Proc. ACM Manag. Data},
  volume = {3},
  number = {1},
  pages = {59:1--59:26},
  urldate = {2025-02-21}
}

@article{zhouNetworkRepresentationLearning2022,
  title = {Network {{Representation Learning}}: {{From Preprocessing}}, {{Feature Extraction}} to {{Node Embedding}}},
  shorttitle = {Network {{Representation Learning}}},
  author = {Zhou, Jingya and Liu, Ling and Wei, Wenqi and Fan, Jianxi},
  year = {2022},
  month = jan,
  journal = {ACM Comput. Surv.},
  volume = {55},
  number = {2},
  pages = {38:1--38:35},
  urldate = {2025-02-28}
}

@article{xuUnderstandingGraphEmbedding2021,
  title = {Understanding {{Graph Embedding Methods}} and {{Their Applications}}},
  author = {Xu, Mengjia},
  year = {2021},
  month = jan,
  journal = {SIAM Review},
  volume = {63},
  number = {4},
  pages = {825--853},
  publisher = {{Society for Industrial and Applied Mathematics}},
  urldate = {2025-02-28}
}

@article{yanHypergraphJointRepresentation2024a,
  title = {Hypergraph {{Joint Representation Learning}} for {{Hypervertices}} and {{Hyperedges}} via {{Cross Expansion}}},
  author = {Yan, Yuguang and Chen, Yuanlin and Wang, Shibo and Wu, Hanrui and Cai, Ruichu},
  year = {2024},
  month = mar,
  journal = {AAAI},
  volume = {38},
  number = {8},
  pages = {9232--9240},
}

@inproceedings{liSHARESystemHierarchical2022,
  title = {{{SHARE}}: A {{System}} for {{Hierarchical Assistive Recipe Editing}}},
  shorttitle = {{{SHARE}}},
  booktitle = {EMNLP},
  author = {Li, Shuyang and Li, Yufei and Ni, Jianmo and McAuley, Julian},
  year = {2022},
  pages = {11077--11090}
}

@inproceedings{chienYouAreAllSet2021,
  title = {You Are {{AllSet}}: {{A Multiset Function Framework}} for {{Hypergraph Neural Networks}}},
  shorttitle = {You Are {{AllSet}}},
  booktitle = {ICLR},
  author = {Chien, Eli and Pan, Chao and Peng, Jianhao and Milenkovic, Olgica},
  year = {2021}
}

@article{gaoHGNNGeneralHypergraph2023,
  title = {{{HGNN}}+: {{General Hypergraph Neural Networks}}},
  shorttitle = {{{HGNN}}+},
  author = {Gao, Yue and Feng, Yifan and Ji, Shuyi and Ji, Rongrong},
  year = {2023},
  month = mar,
  journal = {TPAMI},
  volume = {45},
  number = {3},
  pages = {3181--3199},
  howpublished = {\url{https://github.com/iMoonLab/DeepHypergraph}},
  urldate = {2024-07-09}
}

@misc{report,
year={2025},
author = {},
	title = "Technical Report",
	howpublished = {\url{https://github.com/CyanideCentral/AHNEE/blob/main/SAHE_Technical_Report.pdf}}}

\end{document}